\documentclass{article}
\usepackage{graphicx}
\usepackage{amsmath, calc}
\usepackage{qcircuit}
\usepackage{amsfonts}
\usepackage{caption}
\usepackage{upgreek}
\usepackage[a4paper, margin=0.8in]{geometry}
\usepackage[hidelinks,pagebackref,colorlinks,colorlinks=true, allcolors=violet]{hyperref}
\renewcommand{\backref}[1]{}

\renewcommand{\backrefalt}[4]{%
\ifcase #1 %
\or
[p.\ #2]%
\else
[pp.\ #2]%
\fi}

\usepackage{xcolor}
\usepackage{tabularx}
\usepackage{setspace}
\usepackage{array}
\usepackage{float}
\usepackage{amsthm}
\usepackage{amssymb}
\usepackage{paralist}
\usepackage{authblk}
\usepackage{bbm}
\usepackage{orcidlink}
\usepackage[T1]{fontenc}

\usepackage{cite}
\usepackage[capitalise,nameinlink]{cleveref}
\Crefname{figure}{Figure}{Figures}

\usepackage{authblk}
\usepackage{diagbox}
\usepackage{makecell}

\newcolumntype{N}{>{$\displaystyle}l<{$}} 
\newtheorem{theorem}{Theorem}
\newtheorem{lemma}[theorem]{Lemma}
\newtheorem{corollary}[theorem]{Corollary}%

\theoremstyle{definition}
\newtheorem{definition}[theorem]{Definition}
\allowdisplaybreaks

\def\wt{\mathrm{wt}}
\def\supp{\mathrm{supp}}
\def\iverson#1{\left[#1\right]}
\def\iversontext#1{\left[\mbox{\normalfont #1}\right]}
\def\e{\mathrm{e}}
\def\i{\mathrm{i}}
\def\ket#1{|#1\rangle}
\def\bra#1{\langle#1|}
\def\braket#1#2{\langle#1|#2\rangle}

\newcommand{\doi}[1]{\href{https://doi.org/#1}{doi:\,#1}}
\usepackage{enumitem}
\setlist[enumerate]{label=\alph*)}
\usepackage{tikz}
\newcommand*\circled[1]{\tikz[baseline=(char.base)]{
            \node[shape=circle,draw,inner sep=2pt] (char) {#1};}}

\title{\textbf{Quantum Volunteer's Dilemma}}
\author[1,2,3]{Dax Enshan Koh\,\orcidlink{0000-0002-8968-591X}\thanks{\href{mailto:dax_koh@ihpc.a-star.edu.sg}{dax\_koh@ihpc.a-star.edu.sg}
}}
\author[3,4]{Kaavya Kumar\,\orcidlink{0009-0009-2467-5137}\thanks{\href{mailto:kaavya.kumars@gmail.com}{kaavya.kumars@gmail.com}}}
\author[1,3,5]{Siong Thye Goh\,\orcidlink{0000-0001-7563-0961}\thanks{\href{mailto:gohst2@ihpc.a-star.edu.sg}{gohst2@ihpc.a-star.edu.sg}}}

\affil[1]{\small A*STAR Quantum Innovation Centre (Q.InC), Institute of High Performance Computing (IHPC), Agency for Science, Technology and Research (A*STAR), 1 Fusionopolis Way, \#16-16 Connexis, Singapore 138632, Singapore}
\affil[2]{\small Science, Mathematics and Technology Cluster, Singapore University of Technology and Design, 8 Somapah Road, Singapore 487372, Singapore}
\affil[3]{\small Institute of High Performance Computing (IHPC), Agency for Science, Technology and Research (A*STAR), 1 Fusionopolis Way, \#16-16 Connexis, Singapore 138632, Singapore}
\affil[4]{\small United World College South East Asia East Campus, 1 Tampines Street 73, Singapore 528704, Singapore
}
\affil[5]{\small Singapore Management  University, 81 Victoria St, Singapore 188065, Singapore}

\date{}

\begin{document}

\maketitle

\begin{abstract}
The volunteer's dilemma is a well-known game in game theory that models the conflict players face when deciding whether to volunteer for a collective benefit, knowing that volunteering incurs a personal cost. In this work, we introduce a quantum variant of the classical volunteer’s dilemma, generalizing it by allowing players to utilize quantum strategies. Employing the Eisert–Wilkens–Lewenstein quantization framework, we analyze a multiplayer quantum volunteer's dilemma scenario with an arbitrary number of players, where the cost of volunteering is shared equally among the volunteers. We derive analytical expressions for the players' expected payoffs and demonstrate the quantum game's advantage over the classical game. In particular, we prove that the quantum volunteer's dilemma possesses symmetric Nash equilibria with larger expected payoffs compared to the unique symmetric Nash equilibrium of the classical game, wherein players use mixed strategies. Furthermore, we show that the quantum Nash equilibria we identify are Pareto optimal. Our findings reveal distinct dynamics in volunteer's dilemma scenarios when players adhere to quantum rules, underscoring a strategic advantage of decision-making in quantum settings.
\end{abstract}

\section{Introduction}

The volunteer's dilemma, first introduced by Diekmann in 1985 \cite{diekmann1985volunteers}, is a famous $n$-player game in game theory, wherein players must decide whether or not to volunteer and incur a personal cost for the benefit of the entire group. This dilemma encapsulates numerous real-world social traps \cite{diekmann1986volunteers} and explains various human behaviors \cite{archetti2009volunteers, archetti2009cooperation,lee2009rational,krueger2019vexing,heine2021self}, such as the bystander effect \cite{latane1970unresponsive,panchanathan2013bystander,tutic2014procedurally,thomas2016recursive,kwak2020modeling,mercade2021volunteers}, free-riding \cite{otsubo2008dynamic,hilbe2014cooperation,dineen2021formal}, diffusion of responsibility \cite{darley1968bystander,barron2002private,goeree2017experimental,przepiorka2018heterogeneous}, and the NIMBY (`not in my backyard') syndrome \cite{wolsink1990siting, wolsink1994entanglement}. The dilemma also extends beyond human interactions to non-human animal behaviors such as the transmission of predator information \cite{ archetti2009volunteers, searcy2010evolution, archetti2011strategy,mielke2019snake, steinegger2020laboratory,heifetz2021arabian,van2022male,broom2022game,padget2023guppies}. Animals that volunteer to transmit such information may inadvertently alert the predator and put themselves at risk to benefit the group, as seen in species like sooty mangabeys \cite{mielke2019snake} and guppies \cite{padget2023guppies}. This dilemma also occurs in cancer cells \cite{morsky2018cheater,archetti2019cooperation,manini2022ecology} and microbial populations \cite{pedroso2018impact,patel2019crystal}. For instance, during toxin production in bacteria, each bacterium faces a decision: to either pay the metabolic cost of producing a toxin to reach the necessary threshold for successful invasion of the host, or rely on another bacterium to bear the cost \cite{patel2019crystal}.

Over the last few decades, extensive research has been conducted on the volunteer's dilemma, with theoretical studies characterizing various properties of the game \cite{weesie1994incomplete,tutic2014procedurally,konrad2021volunteers}, 
and experimental investigations exploring the factors influencing volunteering choices and how closely actual behavior aligns with theoretical predictions \cite{diekmann1986volunteers,rapoport1988experiments, murnighan1993volunteer,franzen1994group, wook1997effects,otsubo2008dynamic,przepiorka2013individual,franzen2013volunteer,diekmann2015take,krueger2016expectations,goeree2017experimental,healy2018cost,peuker2019yes,przepiorka2021emergence,villiger2023role}. Various versions of the dilemma have been studied, including Diekmann's original symmetric volunteer's dilemma~\cite{diekmann1985volunteers}, where volunteers all incur the same cost and players all receive the same benefit; the asymmetric volunteer's dilemma \cite{diekmann1993cooperation,weesie1993asymmetry,he2014evolutionary,healy2018cost,guo2023asymmetric}, which allows volunteers to incur different costs and players to receive different benefits; the threshold volunteer's dilemma \cite{chen2013shared,mago2023greed}, which requires a threshold number of volunteers for producing the benefit; and the volunteer's timing dilemma \cite{weesie1993asymmetry,weesie1994incomplete}, where players can observe each other's actions and wait for someone else to volunteer, potentially leading to a state of mamihlapinatapai \cite{dwyer2000mamihlapinatapai}.

One significant volunteer's dilemma variant, introduced by Weesie and Franzen in 1998 \cite{weesie1998cost}, involves cost-sharing among volunteers. In this variant, the cost incurred by the volunteers is divided amongst them, making it more applicable to many real-life scenarios \cite{weesie1998cost,amir2024volunteer}. For example, when countries share a polluted water body, they all benefit if one country decides to bear the cleaning costs; however, the burden can be reduced if multiple countries share the expense \cite{ni2007sharing,dong2012sharing}. Another scenario involves saving a drowning child, where a group rescue is less risky than an individual effort \cite{amir2024volunteer}.

In this study, we introduce a volunteer's dilemma setting governed by the laws of quantum mechanics, where the outcomes of a quantum measurement performed on an entangled quantum system determine whether players volunteer, with each player having the ability to manipulate their local part of the system before the measurement. While there are various ways to realize such a quantum setting \cite{marinatto2000quantum, nawaz2004generalized}, we choose the well-known quantization framework established in quantum game theory by Eisert, Wilkens, and Lewenstein in 1999 \cite{eisert1999quantum}. At the beginning of the game, $n$ players are provided with an entangled $n$-qubit state, with each player holding one qubit. Players then independently decide which quantum operation from a designated set to perform on their respective systems. After players manipulate these systems, a collective measurement of the joint quantum state is performed, revealing whether each player volunteers or not.

For the payoff structure, we use the cost-sharing variant of the volunteer's dilemma à la Weesie and Franzen \cite{weesie1998cost}: if nobody volunteers, all players receive a payoff of 0; if $k>0$ players volunteer, each non-volunteer receives a payoff of $2$ and each volunteer receives a payoff of $2-1/k$. In other words, 2 units are awarded if at least one player volunteers, and the overall cost of volunteering, which is 1 unit, is shared among all the volunteers. The expected payoff of each player is then given by the sum of each possible payoff value, multiplied by its corresponding probability, where the probability is determined by Born's rule \cite{nielsen2010quantum}, which calculates the probability of each measurement outcome based on the pre-measurement joint quantum state of the system.

We assume that the players in our quantum volunteer's dilemma are rational and self-interested \cite{cudd1993game}. Players' decisions are the result of maximizing their own selfish payoff functions, conditional on their beliefs about the other players' optimal behaviors. This could lead to decisions that result in Nash equilibria, from which no player can increase their own payoff unilaterally \cite{neumann2004theory}. The main results of our work are threefold: first, we derive analytical expressions for the players’ expected payoffs, showing that they can be written in terms of sums of products of trigonometric functions. Second, using these analytical results, we prove that our quantum game reveals symmetric Nash equilibria with higher expected payoffs compared to the unique symmetric Nash equilibrium \cite{weesie1998cost} found in the classical game, where players employ mixed strategies. Third, we establish that these symmetric Nash equilibria are Pareto optimal, meaning that no player can improve their payoff without reducing the payoff of another player \cite{pardalos2008pareto}. These features of our quantum game make it both attractive and intriguing, offering insights into how quantum strategies can fundamentally alter the dynamics and outcomes of social dilemmas. 

The rest of our paper is structured as follows. In \cref{sec:related_work}, we survey some related work in classical and quantum game theory, emphasizing their significant progress in recent decades and their relevance to our study. In \cref{sec:preliminaries}, we introduce various game-theoretic concepts that are applicable to both classical and quantum game theory. We then review some results in the classical volunteer's dilemma, both with pure and mixed strategies. In \cref{sec:quantum_volunteer's_dilemma}, we define our quantum volunteer's dilemma game and derive explicit expressions for the payoff functions of the players. We then analyze various quantum strategies and prove that some of these strategies are both Nash equilibria and Pareto optimal. Additionally, we prove that the payoffs achieved at these quantum Nash equilibria surpass those attainable in the corresponding classical game with mixed strategies, showcasing the quantum advantage in this context. Finally, in \cref{sec:conclusion}, we conclude with a summary of our results and an outlook on future research directions.

\subsection{Related work}
\label{sec:related_work}

Game theory serves as a foundational framework across various disciplines, including economics, political science, and social science \cite{neumann2004theory,varoufakis2008game}. It explores strategic interactions amongst multiple players, where the payoffs depend not only on each player's decisions but also on the decisions made by others. Its modern origin can be traced back to von Neumann's seminal work in 1928 \cite{von1928theory}, which established fundamental principles for analyzing decision-making in both cooperative and competitive situations. Two important solution concepts developed in game theory that are central to this study are the Nash equilibrium \cite{nash1950equilibrium}, where no player can increase their payoff by unilaterally changing their actions, and Pareto optimality \cite{pardalos2008pareto}, where it is impossible to make any player better off without simultaneously making another player worse off.

The advent of ideas in quantum computation in the 1990s sparked a new question: what if players had access to quantum strategies? This question was first explored by Meyer \cite{meyer1999quantum} and Eisert, Wilkens, and Lewenstein (EWL) \cite{eisert1999quantum} in 1999, leading to the development of quantum game theory. One notable quantum game is the quantum prisoner's dilemma, which ceases to pose a dilemma if quantum strategies are allowed. Unlike the classical prisoner's dilemma game, where players get trapped in non-Pareto optimal Nash equilibria, quantum strategies enable players to achieve Pareto optimal Nash equilibria \cite{eisert1999quantum,eisert2000quantum,du2002playing,dong2021superiority}.

Since the pioneering works of Meyer and EWL, quantum game theory has expanded significantly, introducing a variety of new quantum games. These include quantum Parrondo’s games \cite{flitney2002quantum, lai2020parrondo}, quantum market games \cite{piotrowski2002quantum}, quantum cooperative games \cite{iqbal2002quantum}, the quantum battle-of-the-sexes game \cite{du2000nash,du2001remark,nawaz2004dilemma,consuelo2020pareto,szopa2021efficiency}, the quantum minority game \cite{benjamin2001multiplayer, chen2004n}, and the quantum chicken game (also known as the quantum hawk-dove game) \cite{eisert2000quantum,szopa2021efficiency}. In many of these games, quantum strategies offer strategic advantages over classical ones. For example, the quantum battle-of-the-sexes game resolves the dilemma of multiple equilibria found in the classical version, leading to a unique solution \cite{nawaz2004dilemma}. Similarly, in the quantum chicken game, players using quantum strategies achieve a unique Nash equilibrium with a higher expected payoff compared to what they would achieve using symmetric classical mixed strategies \cite{eisert2000quantum}. Our work extends this body of research, demonstrating that a quantum advantage also emerges in the volunteer's dilemma, which can be viewed as an $n$-player generalization of the chicken game.

The original EWL framework has been extended to various settings, such as iterated \cite{kay2001evolutionary,przepiorka2021emergence, mukhopadhyay2024repeated}, infinitely repeated \cite{ikeda2020foundation,ikeda2021infinitely}, multiplayer \cite{benjamin2001multiplayer}, multi-choice \cite{du2002mutli}, continuous-variable \cite{li2002continuous}, and Bayesian agent-based \cite{debrota2024quantum} scenarios. In these contexts, the quantum advantage often persists, sometimes giving rise to novel behavior. For example, in the multiplayer setting, it has been shown that such games can possess new forms of quantum equilibrium, where entanglement shared among multiple players enables novel cooperative behavior \cite{benjamin2001multiplayer}.

But what accounts for the advantage observed in quantum games? Various studies have explored the role of quantum resources, such as entanglement \cite{du2001entanglement,du2002entanglement,ozdemir2007necessary,li2014entanglement,mohamed2023quantum} and quantum discord \cite{nawaz2010quantum,wei2017quantum}, in creating this advantage. Additionally, the impact of noise in quantum games has been extensively studied \cite{johnson2001playing,chen2003quantum,flitney2004quantum,shuai2007effect,huang2016quantum,khan2018dynamics,kairon2020noisy,legon2023joint,silva2023maximizing}, revealing that in the presence of noise, a system may lose its quantum characteristics, reverting to classical behavior and thereby diminishing the quantum advantage. More recently, studies have focused on identifying conditions that specify appropriate unitary strategies within the EWL framework, further refining our understanding of quantum strategic interactions \cite{frkackiewicz2016strong,frkackiewicz2022nash,frkackiewicz2017quantum,frkackiewicz2024permissible,frkackiewicz2024permissible_four}.

With the advent of small- and intermediate-scale quantum devices \cite{cheng2023noisy}, various groups have experimentally demonstrated quantum games using platforms such as linear optics \cite{lu2004linear,schmid2010experimental}, nuclear magnetic resonance \cite{du2002experimental,mitra2007experimental}, ion traps \cite{buluta2006quantum}, and superconducting devices, including IBM's quantum computers accessible through the cloud \cite{xu2022experimental}. These implementations serve as important benchmarks for early quantum devices and confirm that they operate as expected. Additionally, an interesting study by Chen and Hogg explored how well humans play quantum games \cite{chen2006how}. Surprisingly, they found that even without formal training in quantum mechanics, participants nearly achieved the payoffs predicted by quantum game theory.

Quantum game theory has found diverse applications across a range of fields, including high-frequency trading \cite{khan2021quantum}, negotiations \cite{szopa2014quantum}, reducing food waste in supply chains \cite{li2021reducing}, sustainable development \cite{he2024promoting,he2024reducing,peng2024promoting}, cooperation between pharmaceutical companies \cite{elgazzar2021coopetition}, and open-access publishing \cite{miriyala2024open}. A significant area of application is wireless communication, where quantum strategies can optimize spectrum-sharing in cognitive radio networks \cite{zabaleta2017quantum}. In traffic flow management, quantum game theory has been used to design optimal routing strategies, helping to reduce congestion and improve efficiency \cite{solmeyer2018quantum}. Furthermore, in the domain of quantum networks, quantum game-theoretic approaches have been employed to enhance entanglement distribution protocols, maximizing fidelity and minimizing latency in quantum communication \cite{dey2023quantum}. Our work on the quantum volunteer's dilemma contributes to this growing body of work by offering a framework that can help resolve social dilemmas or improve network coordination.

For those interested in delving deeper into quantum game theory, we recommend various introductory materials \cite{flitney2002introduction,das2023quantumizing}, lecture notes \cite{eisert2000quantum}, surveys \cite{guo2008survey,huang2018survey,kolokoltsov2019quantum}, and reviews \cite{khan2018quantum,ghosh2021quantum} that cover the field’s history, development, and applications.

\section{Preliminaries}
\label{sec:preliminaries}

\subsection{Game-theoretic concepts}

We begin by laying the groundwork for the concepts central to this paper. At the heart of game theory lies the notion of a \textit{game}, which involves multiple players making strategic decisions. Each player is presented with a set of possible strategies, and their payoff is determined by their payoff function, evaluated based on the combination of strategies chosen by all players. We express this formally as follows. 

\begin{definition}[$n$-player game]
For $n \in \mathbb{Z}^+$, an $n$-\textit{player game (in strategic form)} is a  $2n$-tuple 
\begin{align}
    G = (T_1,T_2,\ldots ,T_n; \$_1, \$_2,\ldots,\$_n),
\end{align} 
where each $T_i$ is a set and each $\$_i : T_1 \times T_2 \times \cdots \times T_n \to \mathbb R$ is a real-valued function.
\end{definition}
In this context, $T_i$ represents the \textit{set of strategies} available to player $i$, and $\$_i $ denotes the \textit{payoff function} for player $i$. The value $\$_i(t_1,\ldots,t_n)$ indicates the payoff received by player $i$ when for each $j \in [n]$, the $j$-th player chooses strategy $t_j \in T_j$. The $n$-tuple $(t_1,t_2,\ldots,t_n) \in T_1 \times T_2 \times \cdots \times T_n$ is called a \textit{strategy profile}, representing a complete specification of the strategies chosen by all players.

Next, we turn our attention to three key solution concepts in game theory. For the following definitions, we consider a strategy profile $s = (s_1,s_2,\ldots,s_n) \in T_1 \times T_2 \times \cdots \times T_n$ in the context of an $n$-player game $G = (T_1,T_2,\ldots ,T_n; \$_1, \$_2,\ldots,\$_n)$.

\begin{definition}
[Nash Equilibrium]
The strategy profile $s = (s_1,\ldots,s_n)$ is a \textit{Nash equilibrium} of $G$ if, for every player $i \in [n]$ and for every alternative strategy $t_i \in T_i \setminus \{s_i\}$, the following condition holds:
\begin{align}
\$_i (s_1, s_2, \ldots , s_{i-1}, s_{i}, s_{i + 1}, \ldots ,s_n) \geq \$_i (s_1, s_2, \ldots , s_{i-1}, t_{i}, s_{i + 1}, \ldots ,s_n). 
\end{align}
\end{definition}
In other words, a Nash equilibrium is a strategy profile where no player can improve their payoff by unilaterally deviating from their current strategy, given that the strategies of all the other players remain unchanged \cite{nash1950equilibrium}.

In the special case where each player has only two possible strategies, 0 and 1, i.e., $T_1 = T_2 = \cdots = T_n =\{0,1\}$, a strategy profile $s \in \{0,1\}^n$ is a Nash equilibrium if and only if for all $i \in [n]$, 
\begin{align}
    \$_i (s) \geq \$_i(s \oplus e_i),
    \label{eq:nash_eqm_two}
\end{align}
where $e_i$ is the binary vector that is zero everywhere except on the $i$-th bit, where it is 1, and $\oplus$ denotes binary addition.

A special type of Nash equilibrium is the symmetric Nash equilibrium, where all players adopt the same strategy.
Formally:

\begin{definition}
[Symmetric Nash equilibrium] 
The strategy profile $s = (s_1,\ldots,s_n)$ is a \textit{symmetric Nash equilibrium} of $G$ if it is a Nash equilibrium and all players choose the same strategy, i.e., $s_1 = s_2 = \cdots = s_n$.
\end{definition}

Finally, we introduce the concept of Pareto optimality.
\begin{definition}[Pareto optimal] 
The strategy profile $s = (s_1,\ldots,s_n)$ is \textit{Pareto optimal} in $G$ if for every player $i\in [n]$ and for every strategy profile $t \in T_1\times\cdots \times T_n$, the following condition holds: 
\begin{align} \mbox{if } \$_i(t) > \$_i(s) \mbox{, then there exists } j \in [n] \setminus \{i\} \mbox{ such that } \$_j(t) < \$_j(s).\end{align} 
\label{def:pareto_optimality}
\end{definition}
In other words, a strategy profile is Pareto optimal if improving one player's payoff necessarily results in a decrease in another player's payoff \cite{pardalos2008pareto}.

To conclude, the above solution concepts are fundamental in game theory as they provide insights into strategic interactions. Ideally, we seek strategy profiles that are both a Nash equilibrium and Pareto optimal, as these concepts align with the ideas of efficiency and fairness in strategic interactions.

\subsection{Classical volunteer's dilemma with pure strategies}
\label{sec:deterministic}

The game we consider in this study is the volunteer's dilemma, specifically the cost-sharing version introduced by Weesie and Franzen \cite{weesie1998cost}. In this version, with $n\geq 2$ players, each player decides between volunteering and abstaining. The payoffs of the game are structured as follows: if no players volunteer, all players receive a payoff of 0. If $k>0$ players volunteer, each non-volunteer receives a payoff of $b$, while each volunteer's payoff is $b-c/k$. In other words, if at least one player volunteers, the reward of $b$ units is distributed to all players, while the total cost of $c$ units incurred by volunteering is divided equally among the volunteers. For simplicity, for the rest of this study, we take $b=2$ and $c=1$.

To denote the players' strategies numerically, we use 1 to represent volunteering and 0 to represent abstaining. With this notation, the above game, when players are restricted to pure classical strategies, is represented by the $2n$-tuple $G_{\mathrm{VD}}^{(n)} = (T_1^{\mathrm{VD}},T_2^{\mathrm{VD}},\ldots ,T_n^{\mathrm{VD}}; \$_1^{\mathrm{VD}}, \$_2^{\mathrm{VD}},\ldots,\$_n^{\mathrm{VD}})$, where $T_i^{\mathrm{VD}} = \{0,1\}$ denotes the set of strategies available to player $i$, and $\$_i^{\mathrm{VD}}: \{0,1\}^n \to \mathbb R$ are the payoff functions. These functions, which correspond to \cite[Eq.~(1)]{weesie1998cost} when $b=2$ and $c=1$, are given by:
\begin{align}
    \$_i^{\mathrm{VD}}(x) = \begin{cases}
         2 - \frac{1}{\wt(x)}, & x_i = 1
        \\
        2 \, [\wt(x) > 0], & x_i = 0.
    \end{cases}
    \label{eq:VD_payoff}
\end{align}
Here, $\wt(x) = \{i \in [n]: x_i=1\}$, the Hamming weight of $x$, denotes the number of players choosing to volunteer. The indicator function, represented by the Iverson bracket $[\wt(x)>0]$, equals 1 if $\wt(x) > 0$ (i.e., if at least one player volunteers) and 0 otherwise. The payoffs of the volunteer's dilemma are summarized in table form in  Table~\ref{tab:volunteers_dilemma}. 

\renewcommand{\arraystretch}{1.4}
\captionsetup[table]{
  format=hang,
  width=0.9\textwidth,
  font=small
}
\captionsetup[figure]{
  format=hang,
  width=0.9\textwidth,
  font=small
}
\begin{table}[!ht]
\centering
\begin{tabular}{ |c|*{7} c|}  
 \hline
 \multicolumn{1}{|c|}{} & \multicolumn{7}{c|}{Number of other players who volunteer} \\ \cline{2-8}
 \makecell{Strategy of player $i$} & 0 & 1 & 2 & $\cdots$ & $k-1$ & $\cdots$ & $n-1$ \\ \hline
 \begin{tabular}{c}$1$ (volunteers)\end{tabular} & \begin{tabular}{c}$2-\frac11=1$\end{tabular} & \begin{tabular}{c}$2-\frac12=\frac32$\end{tabular} & \begin{tabular}{c}$2-\frac13=\frac53$\end{tabular}& $\cdots$ & $2-\frac{1}{k}$ & $\cdots$ & $2-\frac{1}{n}$  \\[0.1cm]
 \begin{tabular}{c}$0$ (abstains)\end{tabular} & 0 & 2 & 2 & $\cdots$ & 2 & $\cdots$ & 2  \\ 
 \hline
\end{tabular}
\caption{Payoff matrix for the $n$-player volunteer's dilemma considered in this study. If player $i$ volunteers, they receive a payoff of $2-\tfrac 1k$ if $k-1$ other players volunteer. If player $i$ does not volunteer (i.e., abstains), their payoff is 0 if no other players volunteer, and 2 if at least one other player volunteers.
}
\label{tab:volunteers_dilemma}
\end{table}

There are exactly $n$ Nash equilibria in the game $G_{\mathrm{VD}}^{(n)}$, each characterized by exactly one player volunteering and every other player abstaining. In other words, the Nash equilibria, which are the strategy profiles $x$ that satisfy \cref{eq:nash_eqm_two}, are precisely those $x \in \{0,1\}^n$ with Hamming weight $\wt(x) = 1$. To prove that the strategy profile where exactly one player volunteers (i.e., $\wt(x)=1$) is a Nash equilibrium, consider the case where player $i$ is the sole volunteer, i.e. $x_i = 1$ and $x_j = 0$ for all $j \neq i$. If player $i$ chooses to switch to abstaining, then their payoff would drop to $0$, as no other player would be volunteering. If another player chooses to volunteer alongside player $i$, then that player's payoff would decrease from $2$ to $2- \frac{1}{2} = \frac 32$. Thus, no player can improve their payoff by unilaterally changing their decision, confirming that these strategy profiles are Nash equilibria. Conversely, the strategy profile where all players choose to abstain (i.e., $\wt(x)=0$) cannot be a Nash equilibrium because any player could increase their payoff from 0 to a positive value by choosing to volunteer. Similarly, if more than one player volunteers (i.e, $\wt(x)>1$), then any of the volunteers could increase their payoff from $2 - \frac{1}{k}$ to $2$ by opting to abstain. Therefore, the $n$ strategy profiles where exactly one player volunteers are the only Nash equilibria in the game.

In the special case when the number of players $n=2$, the volunteer's dilemma simplifies to the well-known game of chicken, often depicted as follows: two drivers approach a narrow bridge from opposite directions. The first driver to turn aside allows the other to cross. If neither turns, they risk a potentially deadly head-on collision, which is the most disastrous outcome for both. Each driver prefers to avoid being labeled as the ``chicken'' by staying on course while hoping the other will swerve \cite{rapoport1966game}. 

In this context, swerving corresponds to the strategy of volunteering, and heading straight (i.e., not swerving) corresponds to abstaining. The payoff structure in the volunteer's dilemma can be mapped as follows: if both players volunteer (i.e., both swerve), they receive a payoff of 1. If one player volunteers (i.e., swerves) while the other abstains (i.e., heads straight), the volunteer receives a payoff of 1 while the abstainer receives a payoff of 2 (benefiting from avoiding a crash while not being called a chicken). If neither player volunteers (i.e., both head straight), they receive a payoff of 0 (equivalent to crashing).

The normal form representation of the chicken game, as a bimatrix, is illustrated in Table \ref{tab:payoff_chicken}.

\begin{table}[h!]
\centering
\begin{tabular}{ |c|c|c| } 
 \hline
\backslashbox{Player 1}{Player 2} & $1$ (swerve) & $0$ (straight) \\ \hline
 $1$ (swerve) & \backslashbox[22mm]{$\frac 32$}{$\frac 32$} & \backslashbox[22mm]{1}{2} \\ \hline
 $0$ (straight) & \backslashbox[22mm]{2}{1} & \backslashbox[22mm]{0}{0} \\ 
 \hline
\end{tabular}
\caption{Payoff matrix for the game of chicken, derived by setting the number of players $n= 2$ in the payoff matrix of the volunteer's dilemma shown in Table \ref{tab:volunteers_dilemma}. Swerving is akin to volunteering, and heading straight is akin to abstaining.}
\label{tab:payoff_chicken}
\end{table}

Now, the above characterization of Nash equilibria specializes to the well-known result that, in the game of chicken, the Nash equilibria occur when one player swerves while the other goes straight. Notably, in both chicken and the more general $n$-player volunteer's dilemma, the Nash equilibria are not symmetric, meaning that players do not all adopt the same strategy. Given the symmetric nature of the payoffs, however, one might expect or prefer symmetric Nash equilibria. To achieve such equilibria, one approach is to modify the game $G_{\mathrm{VD}}^{(n)}$ to allow players to employ mixed strategies.

\subsection{Classical volunteer's dilemma with mixed strategies}
\label{sec:mixed}

Mixed strategies in the context of the volunteer's dilemma involve players deciding to volunteer probabilistically rather than deterministically. Specifically, each player's strategy set in this game is the unit interval $[0,1]$, where a strategy of $\pi \in [0,1]$ indicates that the player will choose to volunteer with probability $\pi$. The payoff function for each player is the expected value of the payoff function from the deterministic game, with payoffs weighted by the probability of choosing to volunteer. This formulation interpolates between the pure strategies discussed in \cref{sec:deterministic}, with $\pi=1$ corresponding to volunteering (strategy $x=1$) and $\pi=0$ corresponding to abstaining (strategy $x=0$).

Formally, the classical volunteer's dilemma with mixed strategies may be represented by the
$2n$-tuple $G_{\mathrm{MVD}}^{(n)} = (T_1^{\mathrm{MVD}},T_2^{\mathrm{MVD}},\ldots ,T_n^{\mathrm{MVD}}; \$_1^{\mathrm{MVD}}, \$_2^{\mathrm{MVD}},\ldots,\$_n^{\mathrm{MVD}})$, where $T_i^{\mathrm{MVD}} = [0,1]$ denotes the set of strategies available to player $i$, and $\$_i^{\mathrm{MVD}}: [0,1]^n \to \mathbb R$ are the payoff functions, defined as
\begin{align}
    \$_i^{\mathrm{MVD}}(\pi_1,\pi_2,\ldots,\pi_n) = 
    \sum_{x \in \{0,1\}^n} \$_i^{\mathrm{VD}}(x) q_{\pi_1,\ldots,\pi_n} (x),
    \label{eq:expected_payoff}
\end{align}
where $\$_i^{\mathrm{VD}}$ are the payoff functions of the (deterministic) volunteer's dilemma given by \cref{eq:VD_payoff} and 
\begin{align}
    q_{\pi_1,\ldots,\pi_n} (x_1,\ldots, x_n) = \prod_{i=1}^n \pi_i^{x_i} (1-\pi_i)^{1-x_i} 
    \label{eq:joint_distribution}
\end{align}
denotes the joint probability that the strategy profile $(x_1,\ldots,x_n)$ is chosen.
The strategies of each player are assumed to be chosen independently and so the joint probability mass function $q_{\pi_1,\ldots,\pi_n}$ factorizes; indeed, \cref{eq:joint_distribution} can be written as 
$q_{\pi_1,\ldots,\pi_n} (x_1,\ldots, x_n) = q_{\pi_1} (x_1) \cdots q_{\pi_n} (x_n)$, where $q_{\pi_i}(1) = \pi_i$ is the probability that player $i$ volunteers, and $q_{\pi_i}(0) = 1- \pi_i$ is the probability that they abstain. We will refer to  $\$_i^{\mathrm{MVD}}(\pi_1,\pi_2,\ldots,\pi_n)$ as the \textit{expected payoff} of player $i$ under mixed strategies.

Unlike the volunteer's dilemma with pure strategies $G_{\mathrm{VD}}^{(n)}$, which has no symmetric Nash equilibria, the mixed-strategy game $G_{\mathrm{MVD}}^{(n)}$ admits a unique symmetric Nash equilibrium. Interestingly, at this equilibrium, players volunteer with a probability that is a root of some univariate degree-$n$ polynomial whose coefficients depend on $n$. More precisely, this result may be stated as follows:

\begin{theorem}[Weesie and Franzen
{\cite[Theorem 1i]{weesie1998cost}}]
    The $n$-player volunteer's dilemma with mixed strategies $G_{\mathrm{MVD}}^{(n)}$ has exactly one symmetric Nash equilibrium $(\alpha_n,\alpha_n,\ldots, \alpha_n)$, where $\alpha_n$ is the unique root in the open interval $(0,1)$ of the degree-$n$ polynomial $g_n$ given by
    \begin{align}
        g_n(\alpha) = (1-\alpha)^{n-1}(2n\alpha+1-\alpha)-1.
    \end{align}
\end{theorem}

\begin{figure}
    \centering
    \includegraphics[width=0.6\linewidth]{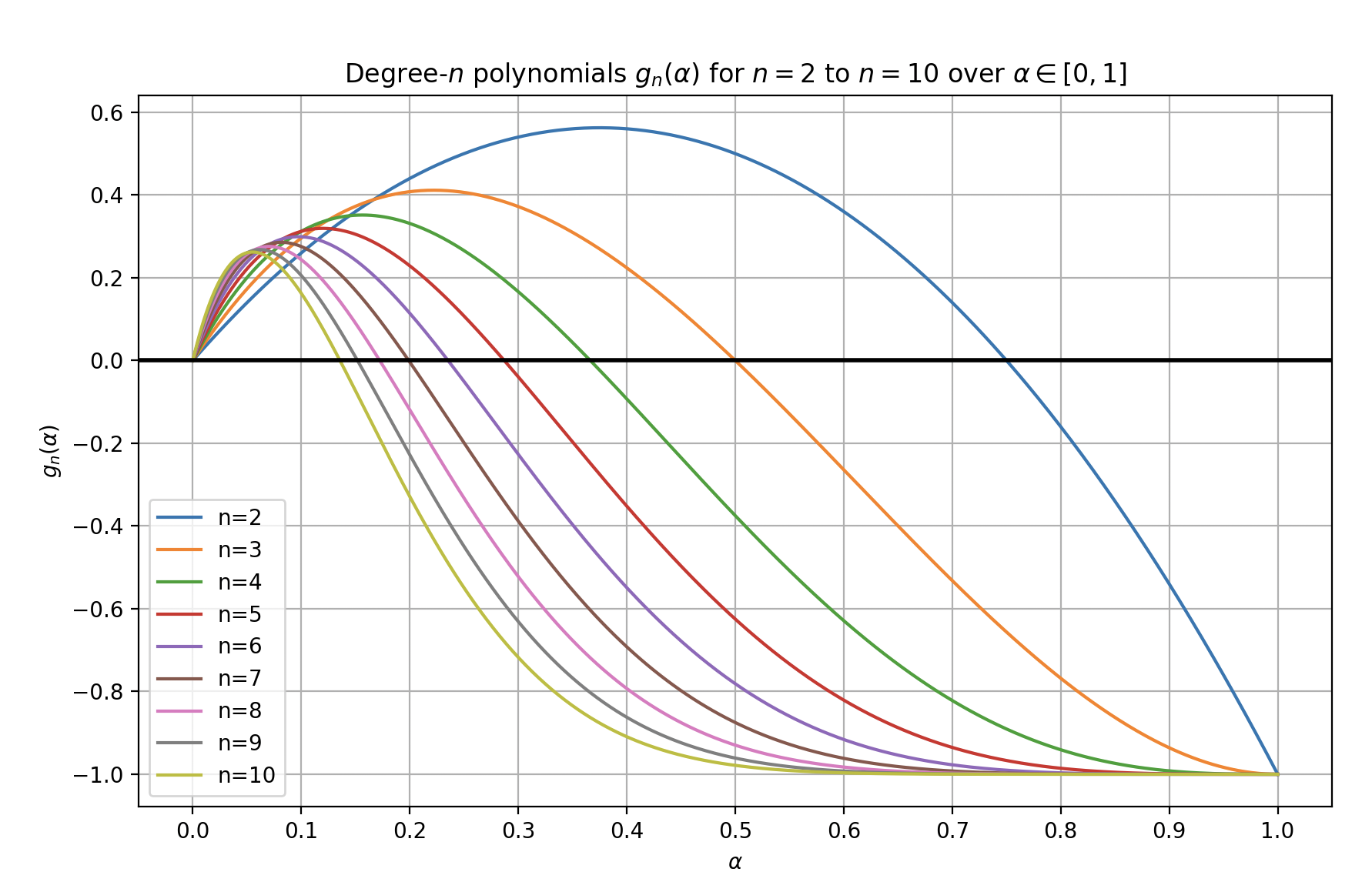}
    \caption{Plots of the degree-$n$ polynomials $g_n(\alpha)$ versus $\alpha$ over the interval $\alpha \in [0,1]$, for $n=2$ to $n=10$. 
    The unique root of $g_n(\alpha)$ within the interval $(0,1)$ is the strategy adopted by each player at the unique symmetric Nash equilibrium in the classical volunteer's dilemma with mixed strategies.}
    \label{fig:gn}
\end{figure}
To illustrate the polynomials $g_n(\alpha)$, we plot them over the interval $[0,1]$ for $n=2$ to $n=10$ in \cref{fig:gn}. As can be seen from this plot, the unique root $\alpha_n$ of $g_n$ in $(0,1)$ decreases as $n$ increases. For large $n$, the root asymptotically behaves as \cite[Theorem 1ii]{weesie1998cost}:
\begin{align}
\alpha_n = \frac{\omega^*}{n}+O(n^{-2}),
\label{eq:approx_alpha_n}
\end{align}
where 
$\omega^*  = -\frac 12 - W_{-1} \! \left(- \frac 1{2 \sqrt{\e}}\right)\approx 1.25643$ is the unique positive solution for $\omega$ of the equation
\begin{align}
\e^\omega= 1 + 2 \omega,
\end{align}
with $W_k(z)$ denoting the $k$-th branch of the Lambert W-function evaluated at $z$, implemented in the Wolfram Language as $\texttt{ProductLog[k,z]}$ \cite{lambert_w_function,product_log}. To illustrate the decreasing trend of $\alpha_n$ and its behavior for large $n$, we present a plot of $\alpha_n$ versus $n$ and compare it with the approximation $\alpha_n \approx \omega^*/n$ in \cref{fig:alpha_n}. 

\begin{figure}
    \centering
    \includegraphics[width=0.6\linewidth]{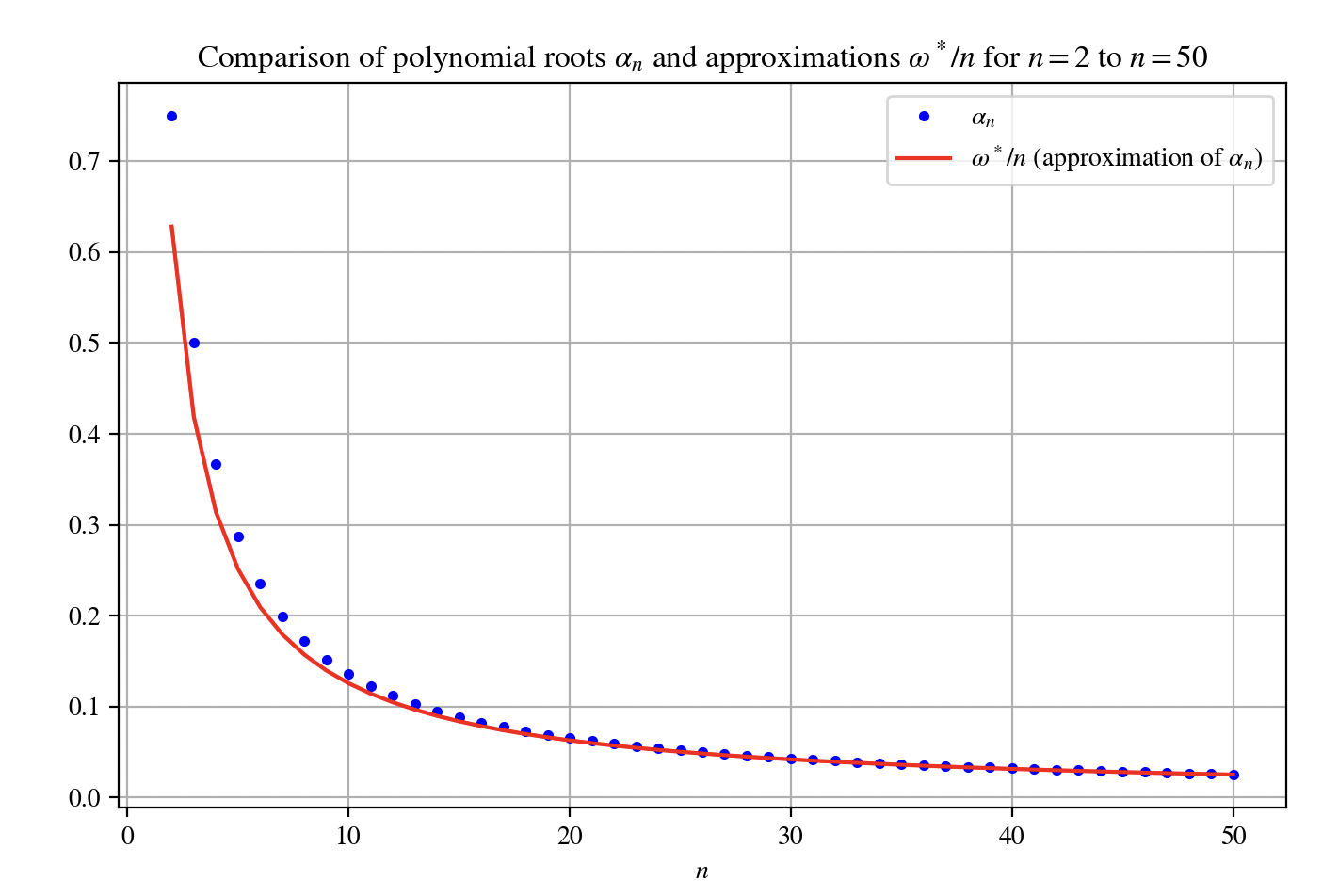}
    \caption{Plots of the roots $\alpha_n$ versus $n$ for $n$ ranging from 2 to 50. The root $\alpha_n$, which is the unique root of the polynomial $g_n$ within the open interval $(0,1)$, represents the strategy adopted by each player at the unique symmetric Nash equilibrium in the classical volunteer's dilemma with mixed strategies.
    The approximation $\omega^*/n$ estimates $\alpha_n$, with an error term of $O(n^{-2})$.
    }
    \label{fig:alpha_n}
\end{figure}

Our next theorem provides an expression for the expected payoff of each player when players choose the unique symmetric Nash equilibrium strategy profile.

\begin{theorem}
    For the $n$-player classical volunteer's dilemma with mixed strategies $G_{\mathrm{MVD}}^{(n)}$, when the unique symmetric Nash equilibrium strategy profile $(\alpha_n,\ldots,\alpha_n)$ is chosen, the expected payoff of each player $i$ defined by \cref{eq:expected_payoff} evaluates to
    \begin{align}
    \$_i^{\mathrm{MVD}}(\alpha_n, \alpha_n, \ldots, \alpha_n) = \left( 2-\frac1n\right)\cdot  \left(1-(1-\alpha_n)^n\right).
    \label{eq:payoff_MVD_symmetricNE}
\end{align}
\end{theorem}
\begin{proof}
Consider the scenario in which every player chooses to volunteer with probability $\alpha_n$, and let $A$ be the event where every player abstains. The probability of the event $A$ is given by
$$\mathbb P(A)=(1-\alpha_n)^n.$$ 
In this case, the conditional expected payoff for each player is 
\begin{align}
    \mathbb E[\$_i^{\mathrm{VD}}(x)|A]= 0,
\end{align} 
as there are no volunteers.

Next, consider the complement event $A^{\mathsf{c}}$, where at least one player volunteers. By the complement rule, the probability of this event is 
$$ \mathbb P(A^{\mathsf{c}})=1- \mathbb P(A)=1-(1-\alpha_n)^n.$$ If the number of volunteers is $k>0$, then the total payoff distributed among all players is $(2 - \frac{1}{k})k + 2(n-k) = 2n-1$. Consequently, given that at least one player volunteers, the conditional expected payoff for each player is 
\begin{align}
    \mathbb E[\$_i^{\mathrm{VD}}(x)|A^{\mathsf{c}}] = \frac{2n-1}{n} = 2-\frac{1}{n}.
\end{align}
Therefore, by the law of total expectation, the expected payoff for each player can be expressed as
\begin{align*}
\$_i^{\mathrm{MVD}}(\alpha_n,\alpha_n,\ldots,\alpha_n) &= \mathbb E[\$_i^{\mathrm{VD}}(x)] 
\\
&= 
\mathbb E[\$_i^{\mathrm{VD}}(x)|A] \cdot \mathbb P(A)+ \mathbb E[\$_i^{\mathrm{VD}}(x)|A^{\mathsf{c}}] \cdot \mathbb P(A^\mathsf{c}) \\
&= 0 \cdot (1-\alpha_n)^n + \left( 2-\frac1n\right)\cdot \left(1-(1-\alpha_n)^n\right)\\
&= \left( 2-\frac1n\right)\cdot  \left(1-(1-\alpha_n)^n\right),
\end{align*}
which completes the proof of the theorem.
\end{proof}

A few remarks about the expected payoff given in \cref{eq:payoff_MVD_symmetricNE} are in order. First, note that the expression $1-(1-\alpha_n)^n < 1$ since $\alpha_n \in (0,1)$. Consequently, the expected payoff of each player is bounded above as follows: 
\begin{align}
\label{eq:payoff_upper_bound}
\$_i^{\mathrm{MVD}}(\alpha_n, \alpha_n, \ldots, \alpha_n) < 2-\frac1n.    
\end{align}

Second, for large $n$, $\alpha_n$ scales as $\frac{\omega^*}n$, with this approximation being valid when neglecting terms of order $n^{-2}$ and smaller (\cref{eq:approx_alpha_n}). Hence, in the limit as $n$ approaches infinity, the expected payoff converges to
\begin{align}
\label{eq:payoff_limit}
    \lim_{n\to \infty}\$_i^{\mathrm{MVD}}(\alpha_n, \alpha_n, \ldots, \alpha_n) &= 
    \lim_{n \to \infty} \left( 2-\frac1n\right)\cdot  \left(1-(1-\alpha_n)^n\right)
    \nonumber\\
    &=
    2  \left( 1 -  \lim_{n \to \infty}\left(1-\alpha_n \right)^n\right) \nonumber\\
    &=
    2  \left( 1 -  \lim_{n \to \infty}\left(1-\frac{w^*}{n}\right)^n\right) \nonumber\\
    &= 2  \left( 1 -  \e^{-w^*}\right)
    \approx 1.43042.
\end{align}

The relational statements in \cref{eq:payoff_upper_bound,eq:payoff_limit} can be contrasted with those for the quantum volunteer's dilemma, which we will introduce in \cref{sec:quantum_volunteer's_dilemma}. In contrast to the strict upper bound in \cref{eq:payoff_upper_bound}, where the classical payoff is strictly less than $2-\frac 1n$, the quantum volunteer's dilemma can achieve symmetric Nash equilibria where the payoff of each player equals $2-\frac 1n$. Moreover, whereas the limit of the payoff in the classical game as $n$ approaches infinity is $\approx\!\!1.43042$, the quantum game's limit is significantly larger, approaching $\lim_{n\to\infty} (2-\frac 1n) = 2$. This increase in payoffs underscores the advantage of the quantum game over its classical counterpart. We compare the payoffs at these symmetric Nash equilibria for both the classical and quantum games in \cref{fig:payoffs_symmetric}.

\begin{figure}
    \centering
    \includegraphics[width=0.7\linewidth]{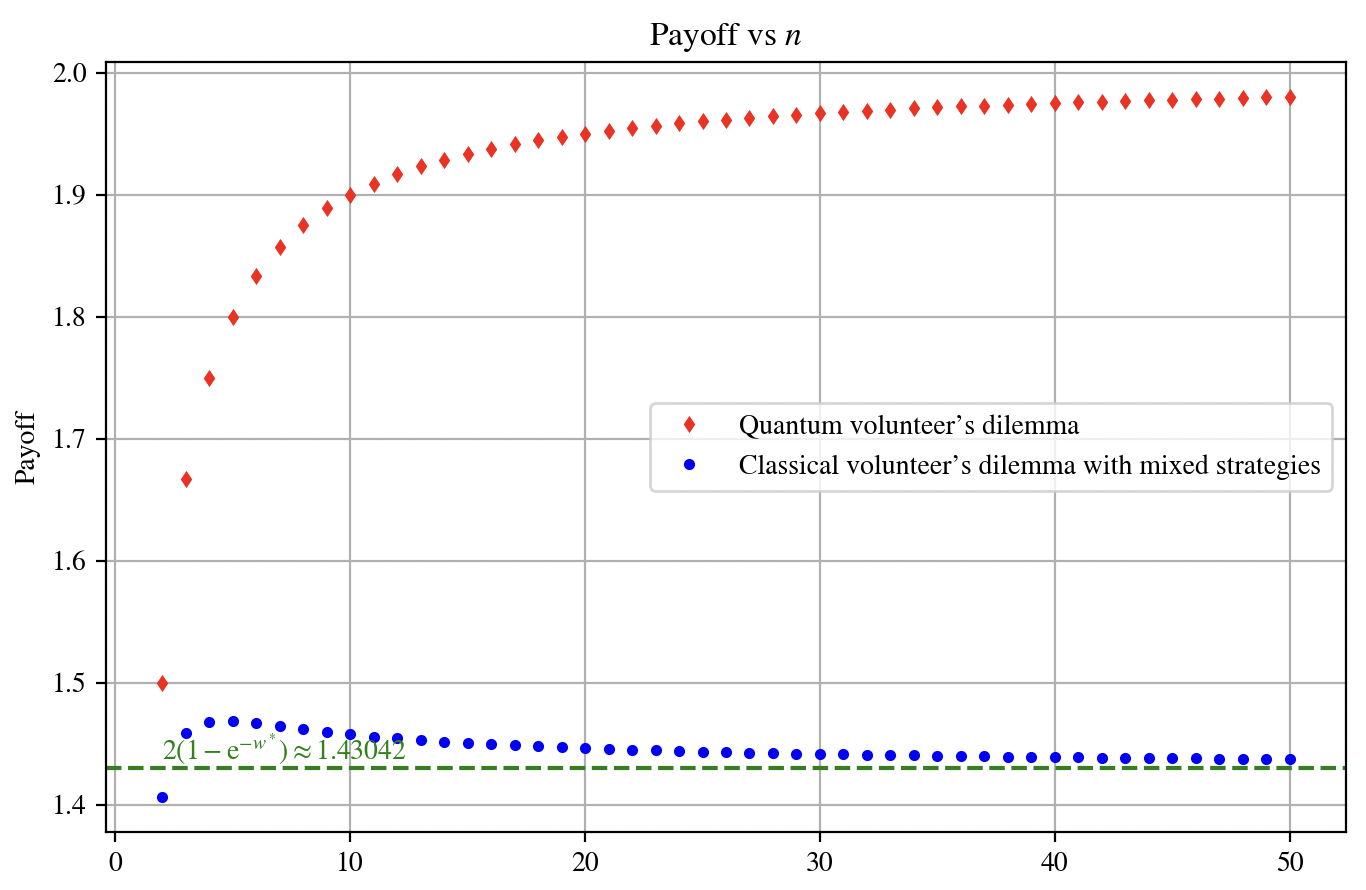}
    \caption{Payoffs at symmetric Nash equilibria for both the quantum volunteer's dilemma and the classical volunteer's dilemma with mixed strategies, as a function of the number of players $n$. The classical payoff asymptotically approaches $2(1-\e^{-w^*})\approx 1.43042$, while the quantum payoff, which is $2-\frac 1n$, asymptotically approaches 2, demonstrating a larger value.
    } \label{fig:payoffs_symmetric}
\end{figure}

\section{Quantum volunteer's dilemma}
\label{sec:quantum_volunteer's_dilemma}

We now introduce the quantum volunteer's dilemma, which extends the classical volunteer’s dilemma by allowing players to use quantum strategies rather than classical ones, following the Eisert--Wilkens--Lewenstein quantization framework \cite{eisert1999quantum}.

\subsection{Game setup}

The quantum volunteer's dilemma game begins with $n$ players receiving an  $n$-qubit entangled state, with each player controlling one qubit. We fix the shared state as $J \ket 0^{\otimes n}$, where $J$ is the $n$-qubit entangling unitary operator
\begin{align}
   J = \e^{-\frac{\i\uppi}{4} Y^{\otimes n}} = \tfrac 1{\sqrt2} (I-\i Y^{\otimes n}),
\end{align}
where $Y= \left(\begin{smallmatrix}
    0 & -\i \\ \i & 0
\end{smallmatrix}\right)$ is the Pauli-Y matrix.

Each player independently selects a quantum operation from a designated set to apply to their respective qubit. We shall fix this set to be the two-parameter family of unitaries $\{ U(\theta, \phi): \theta \in [0,4\uppi), \phi \in [0,2\uppi) \}$,
where
\begin{align}
    U(\theta, \phi) = 
    \begin{pmatrix} 
    \e^{\i \phi} \cos(\frac{\theta}{2}) & \sin(\frac{\theta}{2}) \\
    -\sin(\frac{\theta}{2}) & \e^{-\i \phi} \cos(\frac{\theta}{2})
    \end{pmatrix}.
    \label{eq:two_parameter_unitaries}
\end{align}
Note that choosing a strategy from the above two-parameter strategies amounts to selecting two real parameters $\theta$ and $\phi$. Accordingly, we represent each player's strategy set by $\Theta: = [0,4\uppi) \times [0,2\uppi)$, and say that player $i$'s strategy is $(\theta_i,\phi_i) \in \Theta$ if they select the operation $U(\theta_i,\phi_i)$ to perform on their qubit. Let $\Theta^n$ denote the $n$-fold Cartesian product of $\Theta$ and denote its elements by $(\theta,\phi) := ((\theta_1,\phi_1),\ldots,(\theta_n,\phi_n))$.

After the players apply their quantum operation $U(\theta_i,\phi_i)$, a collective measurement is performed. We fix this measurement to be the $n$-qubit entangling measurement conducted in the basis $\{J\ket k: k \in \{0,1\}^n\}$. Equivalently, this can be achieved by applying the entangling gate $J^\dag$ followed by a computational basis measurement. We associate the $i$-th bit of the measurement outcome $k \in \{0,1\}^n$ with whether player $i$ volunteers, with $k_i = 0$ indicating that they volunteer, and $k_i = 1$ indicating that they abstain. The quantum circuit depicted in \cref{fig:volunteer_dilemma_circuit} provides a visual summary of these steps.

\begin{figure}
    \centering
\begin{align}    \Qcircuit @C=1em @R=.7em {
|0\rangle && \multigate{7}{J} & \gate{U(\theta_1, \phi_1)} & \multigate{7}{J^\dag}& \meter & \cw \\
|0\rangle && \ghost{J}&  \gate{U(\theta_2, \phi_2)} & \ghost{J^\dag} & \meter & \cw \\
|0\rangle && \ghost{J}&  \gate{U(\theta_3, \phi_3)} & \ghost{J^\dag} & \meter & \cw \\
\\
\mspace{260mu} \vdots \mspace{125mu}\vdots
\mspace{125mu}\vdots
\\
\\~\\
|0\rangle && \ghost{J}&  \gate{U(\theta_n, \phi_n)} & \ghost{J^\dag} & \meter & \cw}
\end{align}
\caption{Circuit diagram for the quantum volunteer's dilemma. First, the entangling gate $J$ is applied to the $n$-qubit computational basis state $\ket 0^{\otimes n}$. Each player $i\in [n]$ then applies the single-qubit unitary $U(\theta_i,\phi_i)$ to qubit $i$. Next, the entangling gate $J^\dag$ is applied, followed by a computational basis measurement.}
    \label{fig:volunteer_dilemma_circuit}
\end{figure}
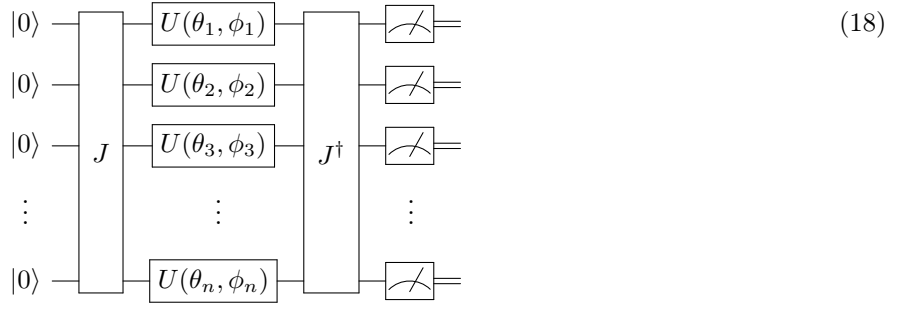

Note that the quantum volunteer's dilemma reduces to the classical volunteer's dilemma with mixed strategies, as discussed in Section~\ref{sec:mixed}, if the players can choose only those $(\theta,\phi)$ in \cref{eq:two_parameter_unitaries} where $\phi = 0$. In this case, volunteering is represented by $U(0,0)=I$, and abstaining represented by $U(\uppi,0)=\left(\begin{smallmatrix} 0 & 1 \\ -1 & 0 \end{smallmatrix}\right)$.

The probability distribution resulting from the measurement may be expressed as follows. If players choose the parameters $\theta = (\theta_1,\ldots,\theta_n) \in [0,4\uppi)^n$ and $\phi = (\phi_1, \ldots, \phi_n) \in [0,2\uppi)^n$, the pre-measurement state before the computational basis measurement is performed is given by
\begin{align}
\label{eq:premeasure} \ket{\psi_f(\theta, \phi)} = J^\dagger \cdot \bigotimes_{i=1}^n U(\theta_i, \phi_i) \cdot J |0\rangle^{\otimes n}. \end{align}

For a string $x\in \{0,1\}^n$, let $\bar x \in \{0,1\}^n$ denote the string where each bit in $x$ is flipped, i.e., $\bar x_i = 1-x_i$. Let $a_{\theta, \phi}(x)$ and $p_{\theta, \phi}(x)$ represent the amplitude and probability, respectively, of observing the binary string $\bar x$ when the state  $|\phi_f\rangle$ is measured in the computational basis. According to Born's rule, the amplitude and probability of measuring $\bar x$ are
\begin{align}
a_{\theta, \phi}(x) &=
\braket{\bar x}{\psi_f(\theta,\phi)}
=
\bra{\bar x} J^\dagger \cdot \bigotimes_{i=1}^n U(\theta_i, \phi_i) \cdot J |0\rangle^{\otimes n} \text{, and }
\label{eq:amplitude_quantum}
\\
    p_{\theta, \phi}(x) 
&=
\left|
a_{\theta, \phi}(x)
    \right|^2
    = 
    \left|
\braket{\bar x}{\psi_f(\theta,\phi)}
    \right|^2
    = \left| \bra{\bar x} J^\dagger \cdot \bigotimes_{i=1}^n U(\theta_i, \phi_i) \cdot J |0\rangle^{\otimes n}
    \right|^2.
    \label{eq:quantum_probability}
\end{align}
Thus, $p_{\theta, \phi}(x)$ is the probability that players in $\supp(x)$ volunteer and players in $[n]\backslash \supp(x)$ abstain, where $\supp(x) = \{i \in [n]: x_i = 1\}$ denotes the support of the string $x \in \{0,1\}^n$.

The expected payoff of player $i$ is then given by
\begin{align}
     \$_i(\theta,\phi) = \sum_{x \in \{0,1\}^n}\$_i^{\mathrm{VD}}(x)p_{\theta,\phi}(x),
     \label{eq:expected_quantum_payoffs}
\end{align}
where $\$_i^{\mathrm{VD}}$ represents the payoff functions in the deterministic volunteer's dilemma as defined in \cref{eq:VD_payoff} and $p_{\theta,\phi}(x)$ is defined in \cref{eq:quantum_probability}.

Formally, the $n$-player quantum volunteer's dilemma can be represented by the $2n$-tuple 
\begin{align}
    G^{(n)}_{\mathrm{QVD}} = (T_1,T_2, \ldots ,T_n;\$_1,\$_2, \ldots ,\$_n),
\end{align}
where $T_1 = T_2 = T_3 = \cdots = T_n = \Theta$ are the strategy sets and 
$\$_i : \Theta^n \to \mathbb{R}$, defined by \cref{eq:expected_quantum_payoffs}, are the payoff functions.

\subsection{Analytical expressions for the payoff functions}

The goal of this section is to simplify the expression for the payoff function in \cref{eq:expected_quantum_payoffs}. To achieve this, we first introduce some auxiliary functions that will help us represent the payoff functions more compactly. Let 
\begin{align}
c_{\theta,\phi}(x) &= \cos\left(\sum_{i:x_i=1} \phi_i\right) \cdot \prod_{i:x_i=0} \sin\left(\frac{\theta_i}{2}\right) \cdot \prod_{i:x_i=1} \cos\left(\frac{\theta_i}{2}\right),
\label{eq:function_c}
\\
s_{\theta,\phi}(x) &= \sin\left(\sum_{i:x_i=0} \phi_i\right) \cdot \prod_{i:x_i=0} \cos\left(\frac{\theta_i}{2}\right) \cdot \prod_{i:x_i=1} \sin\left(\frac{\theta_i}{2}\right).
\label{eq:function_s}
\end{align}
The product of these functions is given by
\begin{align}
    t_{\theta,\phi}(x) &= c_{\theta,\phi}(x)s_{\theta,\phi}(x)
    \nonumber
    \\ 
    &= \frac{1}{2^n} \cos\left(\sum_{i:x_i=1} \phi_i\right) \sin\left(\sum_{i:x_i=0} \phi_i\right) \prod_{i = 1}^n \sin(\theta_i),
    \label{eq:function_t}
\end{align}
where the last line follows from the double-angle identity $\sin(\theta_i )= 2\sin(\theta_i/2) \cos(\theta_i/2)$.

We now present an analytical expression for the payoff functions.

\begin{lemma} \label{lma3} 
The payoff for player $i \in [n]$, as given by \cref{eq:expected_quantum_payoffs}, is given by:
\begin{align}
\label{eq:payoff_simplifed}
     \$_i(\theta,\phi) = 2\sum_{\substack{x \in \{0,1\}^n \\ x \neq 0^n \\ x_i = 0}} p_{\theta,\phi}(x) + \sum_{k = 1}^n \left(2-\frac{1}{k}\right)\sum_{\substack{x \in \{0,1\}^n \\ x_i = 1 \\ \wt(x) = k}} p_{\theta,\phi}(x),
\end{align}
where the probability, as defined in \cref{eq:quantum_probability}, is
\begin{align}
    p_{\theta, \phi}(x)=c_{\theta, \phi}^2 (x) + s_{\theta, \phi}^2 (x)-2(-1)^{\frac{n}2+\wt(x)}t_{\theta, \phi}(x) \iverson{n\in 2 \mathbb{Z}}.
    \label{eq:prob_simplified}
\end{align}
\end{lemma}

\begin{proof}

We begin by deriving \cref{eq:payoff_simplifed}. Starting from \cref{eq:expected_quantum_payoffs}, we obtain

\begin{align}
    \$_i(\theta,\phi) &= 
    \sum_{x \in \{0,1\}^n}\$_i^{\mathrm{VD}}(x)p_{\theta,\phi}(x) \nonumber\\
    &=
    \underbrace{\$_i^{\mathrm{VD}}(0^n)}_{=0}
    p_{\theta,\phi}(0^n) +
\left(
        \raisebox{4mm}{
        $\displaystyle\sum_{
        \underset{
        \substack{x \neq 0^n \\ x_i = 0}
        }{
        x \in \{0,1\}^n
        }
        }
        +
        \displaystyle\sum_{
        \substack{
        x \in \{0,1\}^n \\
        x_i = 1
        }
        }$
        }
        \right)
    \$_i^{\mathrm{VD}}(x)p_{\theta,\phi}(x) \nonumber\\
    &= 
    2 \sum_{
        \underset{
        \substack{x \neq 0^n \\ x_i = 0}
        }{
        x \in \{0,1\}^n
        }
        }
        p_{\theta,\phi}(x)
        +
        \sum_{
        \substack{
        x \in \{0,1\}^n \\
        x_i = 1
        }
        } \left( 2 - \frac{1}{\wt(x)} \right) p_{\theta,\phi}(x)
        \nonumber\\
        &=
        2 \sum_{
        \underset{
        \substack{x \neq 0^n \\ x_i = 0}
        }{
        x \in \{0,1\}^n
        }
        }
        p_{\theta,\phi}(x)
        +
\sum_{k=1}^n
        \sum_{
        \underset{
        \substack{\wt(x) =k \\ x_i = 1}
        }{
        x \in \{0,1\}^n
        }
        }
        \left( 2 - \frac{1}{\wt(x)} \right) p_{\theta,\phi}(x)
        \nonumber\\
        &=
        2\sum_{\substack{x \in \{0,1\}^n \\ x \neq 0^n \\ x_i = 0}} p_{\theta,\phi}(x) + \sum_{k = 1}^n \left(2-\frac{1}{k}\right)\sum_{\substack{x \in \{0,1\}^n \\ x_i = 1 \\ \wt(x) = k}} p_{\theta,\phi}(x).
\end{align}
To prove \cref{eq:prob_simplified}, consider how the quantum state evolves through each step in the circuit shown in \cref{fig:volunteer_dilemma_circuit}. The entangled state that the players receive is
\begin{align}
 J \ket{0}^{\otimes n} = \frac{1}{\sqrt{2}} \left(| 0 \rangle^{\otimes n} - \i^{n+1} | 1\rangle^{\otimes n}\right).   
\end{align}

After the players apply their local operations, the state becomes
\begin{align}
\label{eq:post_local_operation}
    \bigotimes^n_{i = 1} U(\theta_i,\phi_i) \cdot J | 0 \rangle^{\otimes n} = \frac{1}{\sqrt{2}}\left[\bigotimes^n_{i = 1} U(\theta_i,\phi_i) \ket 0 - \i^{n+1} \bigotimes^n_{i = 1} U(\theta_i,\phi_i) \ket 1  \right].
\end{align}

On the other hand, the elements of the measurement basis can be written as
\begin{align}
\label{eq:measurement_basis_J}
    J | \bar{x} \rangle &=  \frac{1}{\sqrt{2}} (I - \i Y^{\otimes n}) | \bar{x}\rangle 
    \nonumber\\
    &= \frac1{2} \left( |\bar{x}\rangle - \i \bigotimes_{j=1}^n Y|\bar{x}_j \rangle \right) \nonumber \\
    &= \frac1{2} \left( |\bar{x}\rangle - \i \bigotimes_{j=1}^n (-1)^{\bar x_j} \i \ket{\bar{x}_j} \right) 
    \nonumber \\
    &= \frac{1}{\sqrt{2}} \left(| \bar{x} \rangle - \i^{n+1} \cdot (-1)^{\wt(\bar{x})} |x \rangle\right).  
\end{align}

Substituting \cref{eq:post_local_operation,eq:measurement_basis_J} into the amplitude expression given by \cref{eq:amplitude_quantum}, we obtain
\begin{align}
    a_{\theta,\phi}(x) &= \frac{1}{\sqrt{2}} \left(\langle \Bar{x} | - (-\i)^{n+1}(-1)^{\wt(\Bar{x})} \langle x | \right) \frac{1}{\sqrt{2}} \left(\bigotimes^n_{i = 1} U(\theta_i,\phi_i) | 0 \rangle - \i^{n+1} \bigotimes^n_{i = 1} U(\theta_i,\phi_i) | 1 \rangle\right)
        \nonumber\\
        &=
        \frac 12 \left(
        \prod_{i=1}^n \bra{\bar x_i} U(\theta_i,\phi_i)\ket 0 - \i^{n+1}
        \prod_{i=1}^n 
\bra{\bar x_i} U(\theta_i,\phi_i)\ket 1 \right.
\nonumber\\
&\quad\qquad \left. -(-\i)^{n+1} (-1)^{\wt(\bar x)} \prod_{i=1}^n \bra{x_i} U(\theta_i,\phi_i)\ket 0
+ (-1)^{\wt(\bar x)} \prod_{i=1}^n \bra{x_i} U(\theta_i,\phi_i)\ket 1
    \right)
\nonumber\\
&= \frac{1}{2} \Bigg[ \prod_{i:x_i = 0} \langle 1 | U(\theta_i,\phi_i) | 0 \rangle \prod_{i:x_i = 1} \langle 0 | U(\theta_i,\phi_i) | 0 \rangle -\i^{n+1} \prod_{i:x_i = 0} \langle 1 | U(\theta_i,\phi_i) | 1 \rangle \prod_{i:x_i = 1} \langle 0 | U(\theta_i,\phi_i) | 1 \rangle \nonumber\\
&\quad\qquad - (-\i)^{n+1} (-1)^{\wt(\Bar{x})} 
\prod_{i:x_i = 0} \langle 0 | U(\theta_i,\phi_i) | 0 \rangle \prod_{i:x_i = 1} \langle 1 | U(\theta_i,\phi_i) | 0 \rangle 
\nonumber\\
&\quad\qquad + (-1)^{\wt(\Bar{x})} \prod_{i:x_i = 0} \langle 0 | U(\theta_i,\phi_i) | 1 \rangle \prod_{i:x_i = 1} \langle 1 | U(\theta_i,\phi_i) | 1 \rangle \Bigg].
\label{eq:amplitude_intermediate}
    \end{align}

Substituting the expressions for the matrix elements $\langle 0 | U(\theta_i,\phi_i) | 0 \rangle = \e^{\i \phi_i} \cos(\frac{\theta_i}{2})$, $\langle 1 | U(\theta_i,\phi_i) | 0 \rangle = -\sin(\frac{\theta_i}{2})$, $\langle 0 | U(\theta_i,\phi_i) | 1 \rangle = \sin(\frac{\theta_i}{2})$, and $\langle 1 | U(\theta_i,\phi_i) | 1 \rangle = \e^{-\i \phi_i} \cos(\frac{\theta_i}{2})$ into \cref{eq:amplitude_intermediate}, we get
\begin{align}
    a_{\theta,\phi}(x) &= \frac{1}{2} \Bigg[ \prod_{i:x_i = 0} \left(-\sin \frac{\theta_i}{2}\right) 
    \cdot
    \prod_{i:x_i = 1} 
    \e^{\i \phi_i} \cos\left(\frac{\theta_i}{2}
    \right) -\i^{n+1} \prod_{i:x_i = 0} \e^{-\i \phi_i} \cos\left(\frac{\theta_i}{2}
    \right) 
    \cdot
    \prod_{i:x_i = 1} \sin\left(\frac{\theta_i}{2}
    \right) \nonumber\\
&\quad\qquad - (-\i)^{n+1} (-1)^{\wt(\Bar{x})} 
\prod_{i:x_i = 0}     \e^{\i \phi_i} \cos\left(\frac{\theta_i}{2}
    \right) \cdot
    \prod_{i:x_i = 1} \left( 
    - \sin \frac{\theta_i}{2}
    \right)
\nonumber\\
&\quad\qquad + (-1)^{\wt(\Bar{x})} \prod_{i:x_i = 0} \sin \frac{\theta_i}{2} 
\cdot
\prod_{i:x_i = 1}     \e^{-\i \phi_i} \cos\left(\frac{\theta_i}{2}
    \right) \Bigg]
    \nonumber\\
&= \frac{1}{2} \Bigg[
(-1)^{\wt(\bar x)} \e^{\i \sum_{i:x_i=1} \phi_i} \prod_{i:x_i=0} \sin \frac{\theta_i}{2} 
\prod_{i:x_i=1} \cos \frac{\theta_i}{2}
-\i^{n+1} \e^{-\i \sum_{i:x_i=0} \phi_i}
\prod_{i:x_i=0} \cos \frac{\theta_i}{2} 
\prod_{i:x_i=1} \sin \frac{\theta_i}{2}
\nonumber\\
&\quad\qquad
-(-\i)^{n+1} (-1)^n
\e^{\i \sum_{i:x_i=0} \phi_i}
\prod_{i:x_i=0} \cos \frac{\theta_i}{2} 
\prod_{i:x_i=1} \sin \frac{\theta_i}{2}
\nonumber\\
&\quad\qquad
+
(-1)^{\wt(\bar x)}
\e^{-\i \sum_{i:x_i=1} \phi_i} \prod_{i:x_i=0} \sin \frac{\theta_i}{2} 
\prod_{i:x_i=1} \cos \frac{\theta_i}{2}
\Bigg]
\nonumber\\
&=
(-1)^{\wt(\bar x)} \cdot \frac 12
\left(
\e^{\i \sum_{i:x_i=1} \phi_i}+\e^{-\i \sum_{i:x_i=1} \phi_i}
\right) \prod_{i:x_i=0} \sin \frac{\theta_i}{2} 
\prod_{i:x_i=1} \cos \frac{\theta_i}{2} 
\nonumber\\
&\quad\qquad
+ \i^{n+1} \cdot
\frac 12
\left(
\e^{\i \sum_{i:x_i=0} \phi_i}-\e^{-\i \sum_{i:x_i=0} \phi_i}
\right) \prod_{i:x_i=0} \cos \frac{\theta_i}{2} 
\prod_{i:x_i=1} \sin \frac{\theta_i}{2} 
\nonumber\\
&=
(-1)^{\wt(\Bar{x})} \cos\left(\sum_{i:x_i = 1} \phi_i\right) \prod_{i:x_i = 0} \sin{\frac{\theta_i}{2}} \prod_{i:x_i = 1} \cos{\frac{\theta_i}{2}} - \i^n \sin\left(\sum_{i:x_i = 0} \phi_i\right) \prod_{i:x_i = 0} \cos{\frac{\theta_i}{2}} \prod_{i:x_i = 1} \sin{\frac{\theta_i}{2}}.
\label{eq:amplitudes_intermediate}
    \end{align}

Squaring these amplitudes gives the probabilities: 
\begin{align}
    p_{\theta,\phi}(x) = \left|(-1)^{\wt(\Bar{x})} \cos\left(\sum_{i:x_i = 1} \phi_i\right) \prod_{i:x_i = 0} \sin{\frac{\theta_i}{2}} \prod_{i:x_i = 1} \cos{\frac{\theta_i}{2}} - \i^n \sin\left(\sum_{i:x_i = 0} \phi_i\right) \prod_{i:x_i = 0} \cos{\frac{\theta_i}{2}} \prod_{i:x_i = 1} \sin{\frac{\theta_i}{2}}
    \right|^2.
\end{align}

The amplitude given by \cref{eq:amplitudes_intermediate} can be either real or have a non-zero imaginary component, depending on whether $n$ is even or odd. Specifically, when $n$ is even, $\i^n$ is real, whereas for odd $n$, $\i^n$ is imaginary. In light of this, when evaluating the probability $p_{\theta,\phi}(x)$, we will consider these two cases separately.

\vspace{0.1cm}
\noindent \underline{Case 1: $n$ is even}. When $n$ is even, $\i^n = (\i^2)^{n/2}=(-1)^{n/2}$, and hence the amplitude given by \cref{eq:amplitudes_intermediate} is real. Therefore, the probability expression simplifies to the square of this real amplitude.
    \begin{align}
        p_{\theta,\phi}(x) &= \left[(-1)^{\wt(\Bar{x})} \cos\left(\sum_{i:x_i = 1} \phi_i\right) \prod_{i:x_i = 0} \sin{\frac{\theta_i}{2}} \prod_{i:x_i = 1} \cos{\frac{\theta_i}{2}} - (-1)^{\frac n2} \sin\left(\sum_{i:x_i = 0} \phi_i\right) \prod_{i:x_i = 0} \cos{\frac{\theta_i}{2}} \prod_{i:x_i = 1} \sin{\frac{\theta_i}{2}}
        \right]^2
        \nonumber\\
        &= \left[(-1)^{\wt(\bar{x})} c_{\theta, \phi}(x) - (-1)^{n/2} s_{\theta, \phi}(x)\right]^2 
        \nonumber\\
        &= c_{\theta, \phi}^2(x) + s_{\theta, \phi}^2(x) - 2(-1)^{\wt(\bar{x})} (-1)^{n/2} c_{\theta, \phi}(x) s_{\theta, \phi}(x) 
        \nonumber\\
        &= c_{\theta, \phi}^2(x) + s_{\theta, \phi}^2(x) - 2(-1)^{\wt(x) + n/2} t_{\theta, \phi}(x),
        \label{eq:prob_even}
    \end{align}
where the last line follows from the fact that $(-1)^{n/2+\wt(\bar x)} = (-1)^{n/2+\wt(x)}$.

\noindent \underline{Case 2: $n$ is odd}. When $n$ is odd, $\i^n = \i(-1)^{\frac{n-1}2}$ is purely imaginary. Therefore, the probability expression simplifies as follows:
\begin{align}
        p_{\theta,\phi}(x) 
        &=\left[\cos\left(\sum_{i:x_i=1} \phi_i\right) \prod_{i:x_i=0} \sin \frac{\theta_i}{2} \prod_{i:x_i=1} \cos \frac{\theta_i}{2} \right]^2 + \left[\sin\left(\sum_{i:x_i=0} \phi_i\right) \prod_{i:x_i=0} \cos\frac{\theta_i}{2} \prod_{i:x_i=1} \sin\frac{\theta_i}{2}\right]^2
         \nonumber \\
         &= c_{\theta, \phi}^2(x) + s_{\theta, \phi}^2(x).
         \label{eq:prob_odd}
    \end{align}

By making use of the Iverson bracket, \cref{eq:prob_even,eq:prob_odd} can be compactly expressed as
\begin{align}
\label{eq:probability_compact}
    p_{\theta,\phi}(x) =
    c_{\theta, \phi}^2 (x) + s_{\theta, \phi}^2 (x)-2(-1)^{\frac{n}2+\wt(x)}t_{\theta, \phi}(x) [n\in 2 \mathbb{Z}],
\end{align}
which completes the proof of the lemma.
\end{proof}

\subsection{Symmetric Nash equilibria for \texorpdfstring{$n \leq 9$}{n≤9}}
\label{sec:strategy_Q}

In this section, we will exhibit the existence of a symmetric Nash equilibrium for the quantum volunteer's dilemma with $n \le 9$ that yields a higher payoff compared to the classical volunteer's dilemma with mixed strategies. As we will show in our next theorem, this symmetric Nash equilibrium occurs when all players choose the strategy $Q:=\left(0, \tfrac{\uppi}n \right) \in \Theta$, i.e., when the strategy profile is $Q^n = (Q,Q,\ldots,Q) = \left((0,\tfrac{\uppi}{n}),(0,\tfrac{\uppi}{n}),\ldots, (0,\tfrac{\uppi}{n})\right) \in \Theta^n$. This corresponds to each player applying the unitary operator
\renewcommand*{\arraystretch}{1}
\begin{align}
    U\left(0,\tfrac{\uppi}{n}\right) = \begin{pmatrix} 
    \e^{\i \uppi/n} & 0 \\
    0 & \e^{-\i \uppi/n}
    \end{pmatrix}
\end{align}
to their respective quantum systems.

Our next theorem also establishes that with the strategy profile $Q^n$, each player will volunteer with probability 1, meaning that the probability distribution in \cref{eq:probability_compact} becomes the degenerate distribution $p_{Q^n}(x) = [x = 1^n]$, resulting in each player's payoff being $\$_a(Q^n) = 2-\frac 1n$. Here $1^n = 11\cdots 1 \in \{0,1\}^n$ denotes the all-ones string.

The key difference in Nash equilibrium behavior between $n\leq 9$ and $n\geq 10$ arises from the following lemma concerning the function $f(x) = \sin^2(\uppi x)-x$. The behavior of $f(x)$ is illustrated in \cref{fig:lemma8}, which helps to visualize the regions where $f(x)$ is positive or negative.
\begin{lemma}
\label{lem:sign_of_sin2pix-x}
Let $x>0$. Then,
    \begin{align}
        \sin^2(\uppi x) - x  \begin{cases}
            < 0, & \mbox{$\mathrm{for}$ } 0 < x \leq \frac{1}{10} \\
            > 0, &  \mbox{$\mathrm{for}$ }\frac{1}{9} \leq x \leq \frac{1}{2}. 
        \end{cases}    
    \end{align}
\end{lemma}

\begin{proof}
    Let $f(x) = \sin^2(\uppi x) - x$. We first evaluate the function at a few points of interest: $f(0)=0$, $f(\frac 1{10})\approx -0.0045 < 0$, $f(\frac19) \approx 0.0043>0$, and $f(0.5) = \frac12>0$.

Since $f$ is continuous, it is sufficient to show that $f$ has exactly one minimum point in the interval $\left(0,\tfrac 1{10}\right)$ and exactly one maximum point in the interval $\left(\tfrac 19, \tfrac 12\right)$. To establish this, we compute the first two derivatives of $f$, solve for the stationary points, and determine whether these points are maxima or minima within the interval $(0, \frac12)$:
\begin{align}
    f'(x) &= \uppi \sin(2\uppi x) - 1 ,
    \nonumber\\
        f''(x) &= 2 \uppi^2 \cos(2\uppi x).
\end{align}

To find the stationary points, set
\begin{align}
f'(x) = 0 &\implies \sin(2\uppi x) = \frac{1}{\uppi} \nonumber\\
&\implies x \in \left(\frac{1}{2\uppi} \arcsin\left(\frac{1}{\uppi}\right) + \mathbb{Z}\right) \cup \left(\frac{1}{2} - \frac{1}{2\uppi}\arcsin\left(\frac{1}{\uppi} \right)+ \mathbb{Z} \right).
\end{align}

The intersection of the above set with the interval $\left(0, \frac12\right)$, which we are interested in, contains only two stationary points, namely
\begin{align}
    x_0 := \frac{1}{2\uppi} \arcsin\left(\frac{1}{\uppi}\right) \approx 0.0515, \qquad  x_1  := \frac{1}{2}-\frac{1}{2\uppi} \arcsin\left(\frac{1}{\uppi}\right) \approx 0.448.
\end{align}

We check that  $f''(x_0) \approx 18.71>0$ and $f''(x_1) \approx -18.71<0$. Thus, $x_0$ is the unique local minimum in the interval $\left(0, \tfrac{1}{10}\right)$ and $x_1$ is the unique local maximum in the interval $\left(\tfrac{1}{9}, \tfrac{1}{2}\right)$. This concludes the proof of the lemma. \end{proof}

\begin{figure}
    \centering
    \includegraphics[width=0.6\linewidth]{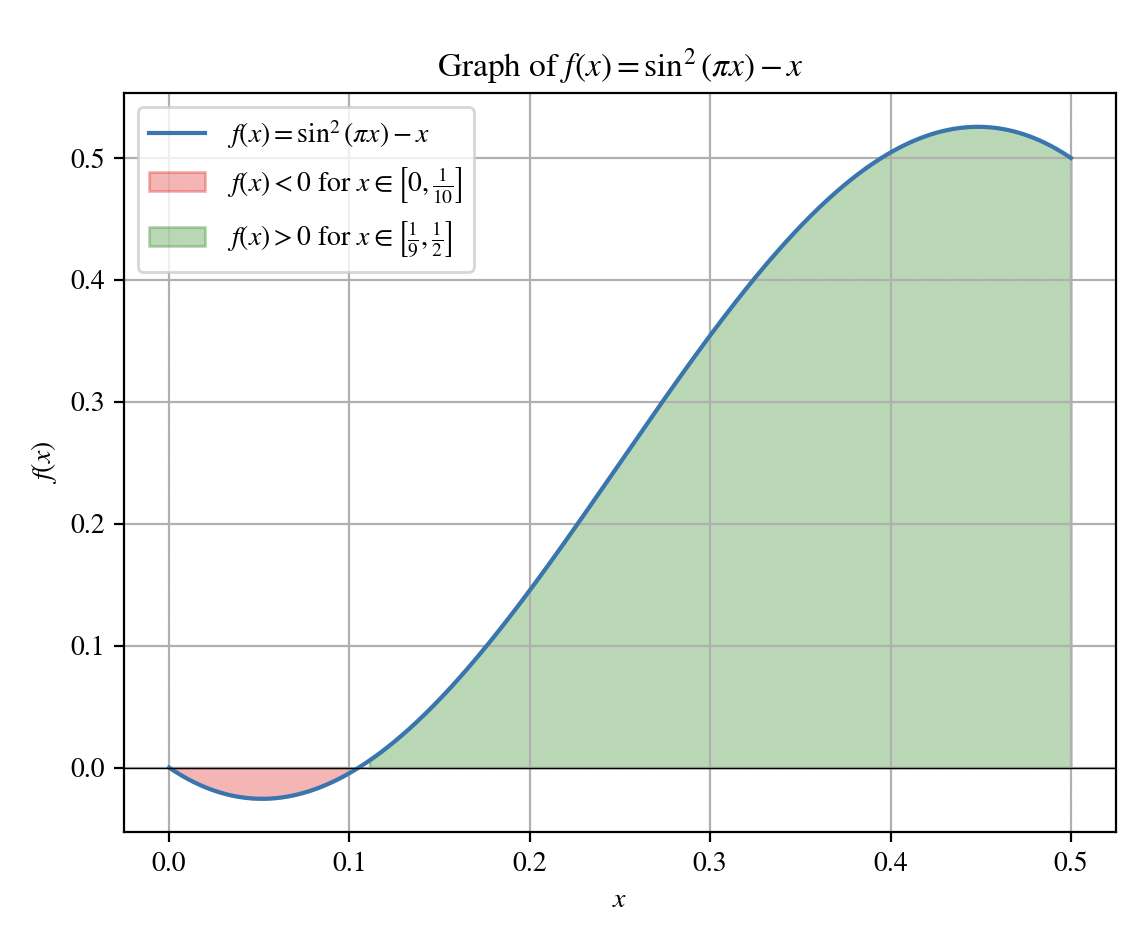}
    \caption{
    Graph of the function $f(x) = \sin^2(\uppi x)-x$ for $x$ in the interval $\left[0,\tfrac 12 \right]$. The graph $f(x)$ is negative for $x \in \left(0,\tfrac 1{10}\right]$ and positive for $x \in \left[\tfrac 19,\tfrac 12\right]$.
    }
    \label{fig:lemma8}
\end{figure}

We now use \cref{lem:sign_of_sin2pix-x} to prove our main theorem about the strategy profile $Q^n$.

\begin{theorem}
\label{thm:QVD_main}
    For the $n$-player quantum volunteer's dilemma $G^{(n)}_{\mathrm{QVD}}$, if the players adopt the strategy profile $Q^n$, then every player volunteers with probability $1$, i.e., the probability distribution is given by
    $p_{Q^n}(x) = [x = 1^n]$, and the corresponding payoff of each player $a\in [n]$ is
    \begin{align}
        \$_a(Q^n) = 2 -\frac 1n.
        \label{eq:payoff_Qn}
    \end{align}
    Moreover,  
    \begin{enumerate} \item if $2 \le n \le 9$, then $Q^n$ is a Nash equilibrium of $G^{(n)}_{\mathrm{QVD}}$, and 
    \item if $n \ge 10$, then $Q^n$ is not a Nash equilibrium of $G^{(n)}_{\mathrm{QVD}}$. \end{enumerate}
\end{theorem}

\begin{proof}
To calculate the probability that each player volunteers, we use the expression in \cref{eq:prob_simplified} derived in \cref{lma3}. At the strategy profile $(\theta,\phi) = Q^n$, the functions defined in \cref{eq:function_c}, \cref{eq:function_s}, and \cref{eq:function_t} simplify as follows:
\begin{align}
    c_{Q^n}(x) &= \cos\left(\sum_{i:x_i=1} \frac{\uppi}{n}\right) \prod_{i:x_i=0} \sin\left(\frac{0}{2}\right) \prod_{i:x_i=1} \cos\left(\frac{0}{2}\right) 
    \nonumber\\
    &=
    \cos\left(\wt(x) \frac{\uppi}{n}\right) \prod_{i:x_i=0} 0
    \nonumber\\
    &=
    \cos\left(\wt(x) \frac{\uppi}{n}\right) \iverson{\forall i: x_i = 1}
    \nonumber\\
    &=
    \cos\left(\wt(x) \frac{\uppi}{n}\right) \iverson{\wt(x) = n}
    \nonumber\\
    &=
    \cos(\uppi)\iverson{\wt(x) = n} \nonumber\\
    &= 
    - \iverson{x = 1^n},
    \\
    s_{Q^n}(x) &= \sin\left(\sum_{i:x_i=0} \frac{\uppi}{n}\right) \prod_{i:x_i=0} \cos\left(\frac{0}{2}\right) \prod_{i:x_i=1} \sin\left(\frac{0}{2}\right) 
    \nonumber\\
    &=
    \sin\left(\wt(\bar x) \frac{\uppi}{n}\right) \prod_{i:x_i=1} 0
    \nonumber\\
    &=
    \sin\left(\wt(\bar x) \frac{\uppi}{n}\right) \iverson{\forall i: x_i = 0}
    \nonumber\\
    &=
    \sin\left(\wt(\bar x) \frac{\uppi}{n}\right) \iverson{\wt(\bar x) = n}
    \nonumber\\
    &=
    \sin(\uppi)\iverson{\wt(\bar x) = n} \nonumber\\
    &= 
    0,
    \\
    t_{Q_n}(1^n) &= c_{Q_n}(1^n)s_{Q_n}(1^n) = 0. 
\end{align}

Hence, $p_{Q^n}(x)= [x = 1^n]$, i.e., $p_{Q^n}(1^n) = 1$. In other words, every player volunteers with probability $1$. Substituting this expression into \cref{eq:payoff_simplifed}
yields the following payoff for each player $a \in [n]$:
\begin{align}
    \$_a(Q^n) = 
    \sum_{k=1}^n \left(2-\frac 1k\right) \left|\{ x \in \{0,1\}^n: x_a =1 \mbox{ and } x = 1^n \mbox{ and } \wt(x) = k\}\right|.
\end{align}
But the set $\left\{ x \in \{0,1\}^n: x_a =1 \mbox{ and } x = 1^n \mbox{ and } \wt(x) = k\right\}$ is empty unless $k=n$, in which case it is equal to the singleton set $\{1^n\}$. Hence,
\begin{align}
    \$_a(Q^n) = \sum_{k=1}^n \left(2-\frac 1k\right) \delta_{k,n}  = 2-\frac 1n,
\end{align}
where $\delta_{k,n}$ is the Kronecker delta.

Next, to determine whether $Q^n$ is a Nash equilibrium, we examine the scenario where a specific player $a \in [n]$ potentially deviates from $Q$, while all other players $i \neq a$ maintain the strategy $Q$. We begin by calculating the payoff for player $a$ under these conditions.
\begin{align}
&c_{\theta, \phi}(x) \Big|_{(\theta_i, \phi_i)=Q, \forall i \ne a} 
\nonumber\\&= \cos\left(\sum_{i:x_i=1} \phi_i\right) \cdot \prod_{i:x_i=0} \sin \frac{\theta_i}{2} \cdot \prod_{i:x_i=1} \cos \frac{\theta_i}{2} \Bigg|_{
\theta_i =0, \phi_i = \tfrac\uppi n, \forall i \neq a
} \nonumber\\
&= \cos\left(x_a\phi_a+\sum_{i:x_i=1, i \ne a} \phi_i\right) \sin^{1-x_a}\left( \frac{\theta_a}2 \right)\cos^{x_a}\left( \frac{\theta_a}2\right)\prod_{\substack{i:x_i=0 \\ i \ne a}}\sin\frac{\theta_i}{2} \prod_{\substack{i:x_i=1\\ i \ne a}} \cos\frac{\theta_i}{2} \Bigg|_{
\theta_i =0, \phi_i = \tfrac\uppi n, \forall i \neq a
}
\nonumber\\
&= \cos\left(x_a\phi_a+\sum_{i:x_i=1, i \ne a} \frac{\uppi}n \right) \sin^{1-x_a}\left( \frac{\theta_a}2 \right)\cos^{x_a}\left( \frac{\theta_a}2\right) \iverson{\forall i \ne a, x_i =1} \nonumber\\ 
&= \cos\left(x_a\phi_a+\frac{n-1}{n}\uppi\right) \sin^{1-x_a}\left( \frac{\theta_a}2 \right)\cos^{x_a}\left( \frac{\theta_a}2\right) \iverson{\forall i \ne a, x_i =1} \nonumber\\
&= -\cos\left( \frac{\uppi}n-x_a\phi_a\right) \sin^{1-x_a}\left( \frac{\theta_a}2 \right)\cos^{x_a}\left( \frac{\theta_a}2\right) \iverson{\forall i \ne a, x_i =1} \nonumber\\
&=\begin{cases} -\cos \left( \frac{\uppi}n\right)\sin\left( \frac{\theta_a}2\right), & \mbox{ for } x_a=0, x_i=1, \forall i \ne a \\
-\cos\left(\frac{\uppi}n-\phi_a \right)\cos\left( \frac{\theta_a}2\right), & \mbox{ for } x_a=1, x_i=1, \forall i \ne a \\ 0, & \text{ otherwise.}\end{cases}
\end{align}
Similarly,
\begin{align}
&s_{\theta, \phi}(x) \Big|_{(\theta_i, \phi_i)=Q, \forall i \ne a} \nonumber\\&= \sin\left(\sum_{i:x_i=0} \phi_i\right) \prod_{i:x_i=0} \cos \frac{\theta_i}{2} \prod_{i:x_i=1} \sin \frac{\theta_i}{2} \Bigg|_{\theta_i =0, \phi_i = \tfrac\uppi n, \forall i \neq a} \nonumber\\
&= \sin\left(\bar{x}_a\phi_a+\sum_{i:x_i=0, i \ne a} \phi_i\right) \cos^{1-x_a}\left( \frac{\theta_a}2 \right)\sin^{x_a}\left( \frac{\theta_a}2\right)\prod_{\substack{i:x_i=0\\ i \ne a}} \cos \frac{\theta_i}{2} \prod_{\substack{i:x_i=1\\ i \ne a}} \sin\frac{\theta_i}{2} \Bigg|_{\theta_i =0, \phi_i = \tfrac\uppi n, \forall i \neq a} \nonumber\\
&= \sin\left(\bar{x}_a\phi_a+\sum_{i:x_i=0, i \ne a} \frac\uppi n\right) \cos^{1-x_a}\left( \frac{\theta_a}2 \right)\sin^{x_a}\left( \frac{\theta_a}2\right) \iverson{\forall i \ne a, x_i =0} \nonumber\\ 
&= \sin\left(\bar{x}_a\phi_a+\frac{n-1}{n}\uppi\right) \cos^{1-x_a}\left( \frac{\theta_a}2\right)\sin^{x_a}\left( \frac{\theta_a}2\right) \iverson{\forall i \ne a, x_i =0} \nonumber\\
&= \sin\left( \frac{\uppi}n-\bar{x}_a\phi_a\right) \cos^{1-x_a}\left( \frac{\theta_a}2 \right)\sin^{x_a}\left( \frac{\theta_a}2\right) \iverson{\forall i \ne a, x_i =1} \nonumber\\
&=\begin{cases} \sin \left( \frac{\uppi}n-\phi_a\right)\cos\left( \frac{\theta_a}2\right), & \mbox{ for } x_a=0, x_i=0, \forall i \ne a \\
\sin\left(\frac{\uppi}n \right)\sin\left( \frac{\theta_a}2\right), & \mbox{ for } x_a=1, x_i=0, \forall i \ne a \\ 0, & \text{ otherwise.}\end{cases}
\end{align}

Since $c_{\theta, \phi}(x) \Big|_{(\theta_i, \phi_i)=Q, \forall i \ne a}$ and $s_{\theta, \phi}(x) \Big|_{(\theta_i, \phi_i)=Q, \forall i \ne a}$ have disjoint supports, it follows that their product $t_{\theta, \phi}(x) \Big|_{(\theta_i, \phi_i)=Q, \forall i \ne a}$ vanishes. Hence, the probability of observing the outcome $x$ is 
\begin{align}
p_{\theta, \phi}(x) \Big|_{(\theta_i, \phi_i)=Q, \forall i \ne a}
&= \cos^2\left( \frac{\uppi}{n}-x_a\phi_a \right)\sin^{2(1-x_a)}\left(\frac{\theta_a}2 \right)\cos^{2x_a}\left( \frac{\theta_a}2\right)
\iverson{\forall i \ne a: x_i=1}
\nonumber\\&\quad + \sin^2\left( \frac{\uppi}{n}-\bar{x}_a\phi_a \right)\cos^{2(1-x_a)}\left(\frac{\theta_a}2 \right)\sin^{2x_a}\left( \frac{\theta_a}2\right)\iverson{\forall i \ne a: x_i=0} \nonumber\\
&= \begin{cases} \cos^2\left( \frac{\uppi}n\right)\sin^2\left( \frac{\theta_a}2\right), & \text{ for } x_a=0, x_i=1, \forall i \ne a 
\ (\text{i.e.,}~ x=11\ldots1\underset{\underset{a}{\uparrow}}{0}11\ldots 1) 
\\
\cos^2\left( \frac{\uppi}n-\phi_a\right)\cos^2\left( \frac{\theta_a}2\right), & \text{ for } x_a=1, x_i=1, \forall i \ne a 
\ (\mbox{i.e.}~ x=1^n) 
\\
\sin^2\left( \frac{\uppi}n-\phi_a\right)\cos^2\left( \frac{\theta_a}2\right), & \text{ for } x_a=0, x_i=0, \forall i \ne a 
\ (\mbox{i.e.,}~ x=0^n) 
\\
\sin^2\left( \frac{\uppi}n\right)\sin^2\left( \frac{\theta_a}2\right), & \text{ for } x_a=1, x_i=0, \forall i \ne a 
\ (\mbox{i.e.,}~ x=00\ldots0\underset{\underset{a}{\uparrow}}{1}0\ldots 00)
\\
0, & \text{ otherwise,}
\end{cases}
\end{align}
where the upward arrows indicate that the 0 or 1, respectively, is located at the $a$-th position in the binary string.

Therefore, the payoff of player $a$ is
\begin{align*}
\$_a(\theta, \phi)\Big|_{(\theta_i, \phi_i)=Q, \forall i \ne a} &= 2\sum_{\substack{x \in \{0,1\}^n \\ x \neq 0^n \\ x_a = 0}} p_{\theta,\phi}(x) + \sum_{k = 1}^n \left(2-\frac{1}{k}\right)\sum_{\substack{x \in \{0,1\}^n \\ x_a = 1 \\ \wt(x) = k}} p_{\theta,\phi}(x)\Bigg|_{(\theta_i, \phi_i)=Q,  \forall i \ne a} \\
&= 2p_{\theta, \phi}(1\ldots 1\underset{\underset{a}{\uparrow}}{0}1\ldots1) + \left(2-\frac11\right)p_{\theta, \phi}(0,\ldots 0\underset{\underset{a}{\uparrow}}{1}0\ldots 0) + \left(2-\frac1n \right) p_{\theta, \phi}(1^n)\Bigg|_{(\theta_i, \phi_i)=Q, \forall i \ne a} \\
&= 2 \cos^2 \left( \frac{\uppi}n\right)\sin^2 \left(\frac{\theta_a}2 \right) + \sin^2\left( \frac{\uppi}n\right)\sin^2\left(\frac{\theta_a}2 \right) + \left( 2-\frac1n\right)\cos^2\left(\frac{\uppi}n-\phi_a \right)\cos^2\left( \frac{\theta_a}2\right)\\
&=\left(1 + \cos^2 \frac{\uppi}{n}\right) \sin^2 \frac{\theta_a}{2} + \left(2-\frac{1}{n}\right) \cos^2\left(\frac{\uppi}{n} - \phi_a\right) \cos^2 \frac{\theta_a}{2}.
\end{align*}

\begin{enumerate}
    \item Consider the case where $2 \le n \le 9$. For any $a\in [n]$, we will show that the payoff of the $a$-th player, when all other players adopt the strategy $Q$, is bounded above by $2-\frac 1n = \$_a(Q^n)$. Indeed,

\begin{align*}
    \$_a(\theta, \phi)\Big|_{(\theta_i, \phi_i)=Q, \forall i \ne a} &=  \left(1 + \cos^2 \frac{\uppi}{n} \right)\sin^2 \frac{\theta_a}{2} + \left(2-\frac{1}{n}\right) \underbrace{\cos^2\left(\frac{\uppi}{n} - \phi_a\right)}_{\leq 1} \cos^2\frac{\theta_a}{2} \\
    & \le \left(1 + \cos^2 \frac{\uppi}{n} \right) \sin^2 \frac{\theta_a}{2} + \left(2-\frac{1}{n}\right) \left( 1-\sin^2 \frac{\theta_a}{2} \right) \\
    &= 2-\frac1n - \left( 1-\cos^2\frac{\uppi}n-\frac1n\right) \sin^2 \left( \frac{\theta_a}2\right)
    \\&= 2-\frac1n  -  \underbrace{\left( \sin^2 \frac{\uppi}n - \frac1n
    \right)
    }_{>0}
    \underbrace{\sin^2\left( \frac{\theta_a}2\right)}_{\geq 0} \\
    &\le 2- \frac1n = \$_a(Q^n),
\end{align*}
where the inequality on the last line follows from Lemma~\ref{lem:sign_of_sin2pix-x} with $x=\frac 1n$. Therefore, the strategy $Q$ provides the optimal payoff for player $a$ when all other players also choose the strategy $Q$. Since the same reasoning applies to every player $a$, the strategy profile $Q^n$ is a Nash equilibrium.

\item Consider the case where $n\geq 10$. For any player $a\in [n]$, we will demonstrate that if all other players adopt the strategy $Q$, there exists an alternative strategy for player $a$ that results in a higher payoff than if they were to choose $Q$, where the payoff from choosing $Q$ would be $2-\frac1n = \$_a (Q^n)$. This shows that $Q^n$ cannot be a Nash equilibrium.

Indeed, an example of an alternative strategy is the choice $(\theta_a, \phi_a) =(\uppi, \frac{\uppi}n)$. With this strategy, the payoff for player $a$ evaluates to
\begin{align}
\$_a(\theta, \phi) \Bigg|_{\substack{(\theta_i, \phi_i)=Q, \forall i \ne a, \\ \theta_a=\uppi, \phi_a =\uppi/n}} &= \left(1 + \cos^2 \frac{\uppi}{n}
\right) \sin^2 \frac{\uppi}{2} + \left(2-\frac{1}{n}\right) \cos^2\left(\frac{\uppi}{n} - \frac{\uppi}{n}\right) \cos^2 \frac{\uppi}{2} \nonumber\\
&= 1 + \cos^2 \left( \frac{\uppi}n\right) \nonumber\\
&= 2-\sin^2 \left(\frac{\uppi}n \right) \nonumber\\
&> 2-\frac1n = \$_a(Q^n),
\end{align}
where the inequality on the last line follows from Lemma \ref{lem:sign_of_sin2pix-x}. Hence $Q^n$ is not a Nash equilibrium for $n\geq 10$.
\end{enumerate}
By combining the results above, we conclude that $Q^n$ is a Nash equilibrium if and only if $n\leq 9$.
\end{proof}

\subsection{Symmetric Nash equilibria for even \texorpdfstring{$n$}{n}}
\label{sec:strategy_A}

In this section, we shift our focus from the previously discussed symmetric Nash equilibria $Q^n$ to a new family of symmetric Nash equilibria for cases with an even number of players. In this new strategy profile, each player adopts the strategy $A=(0, \frac{\uppi}2) \in \Theta$, resulting in a collective strategy profile $A^n = (A,\ldots,A) \in \Theta^n$. This choice corresponds to each player applying the unitary operator
\begin{align}
    U\left(0,\tfrac{\uppi}{2}\right) = \begin{pmatrix} 
    \i & 0 \\
    0 & -\i 
    \end{pmatrix}= \i Z
\end{align}
to their respective quantum systems, where $Z = \left(\begin{smallmatrix}
    1 & 0 \\ 0 & -1
\end{smallmatrix}\right)$ is the Pauli-Z matrix. Like $Q^n$, we will show that this new strategy profile yields a payoff of $2-\frac 1n$, surpassing the payoff of the symmetric Nash equilibrium in the classical volunteer's dilemma.

\begin{theorem}
\label{theorem10}
In the $n$-player quantum volunteer's dilemma $G^{(n)}_{\mathrm{QVD}}$, if the players adopt the strategy profile $A^n$, then every player volunteers with probability $1$ when $n$ is even, while every player abstains with probability $1$ when $n$ is odd, i.e.,
\begin{align}
    p_{A^n}(x) = \iversontext{$n$ even}\iverson{x=1^n} + \iversontext{$n$ odd}\iverson{x=0^n},
    \label{eq:probability_An}
\end{align}
and the corresponding payoff of each player $a\in [n]$ is
    \begin{align}
        \$_a(A^n) = 
        \left(2 -\frac 1n \right)\iversontext{$n$ even}.
        \label{eq:payoff_An}
    \end{align}
Moreover,
\begin{enumerate} \item if $n$ is even, then $A^n$ is a Nash equilibrium of $G^{(n)}_{\mathrm{QVD}}$, and
\item if $n$ is odd, then $A^n$ is not a Nash equilibrium of $G^{(n)}_{\mathrm{QVD}}$. 
\end{enumerate}
\end{theorem}

\begin{proof}
To calculate the probability that each player volunteers, we use the expression in \cref{eq:prob_simplified} derived in \cref{lma3}. At the strategy profile $(\theta,\phi) = A^n$, the functions defined in \cref{eq:function_c}, \cref{eq:function_s}, and \cref{eq:function_t} simplify as follows:
\begin{align}
    c_{A^n}(x) &= \cos\left(\sum_{i:x_i=1} \frac{\uppi}{2}\right) \prod_{i:x_i=0} \sin\left(\frac{0}{2}\right) \prod_{i:x_i=1} \cos\left(\frac{0}{2}\right) 
    \nonumber\\
    &=
    \cos\left(\wt(x) \frac{\uppi}{2}\right) \prod_{i:x_i=0} 0
    \nonumber\\
    &=
    \cos\left(\wt(x) \frac{\uppi}{2}\right) \iverson{\forall i: x_i = 1}
    \nonumber\\
    &=
    \cos\left(\wt(x) \frac{\uppi}{2}\right) \iverson{\wt(x) = n}
    \nonumber\\
    &=
    \cos\left(\frac{n \uppi}{2}\right)\iverson{\wt(x) = n} \nonumber\\
    &= \iverson{\wt(x) = n} \times \begin{cases}  0, & \mbox{for odd } n \\ 1, & \mbox{for } n \equiv 0 \pmod{4} \\ -1, & \mbox{for } n \equiv 2 \pmod{4} \end{cases}\nonumber\\
    &= 
    (-1)^{n/2} \iversontext{$n$ even} \iverson{x = 1^n}.
\end{align}
Similarly,
\begin{align}
    s_{A^n}(x) &= \sin\left(\sum_{i:x_i=0} \frac{\uppi}{2}\right) \prod_{i:x_i=0} \cos\left(\frac{0}{2}\right) \prod_{i:x_i=1} \sin\left(\frac{0}{2}\right) 
    \nonumber\\
    &=
    \sin\left(\wt(\bar x) \frac{\uppi}{2}\right) \prod_{i:x_i=1} 0
    \nonumber\\
    &=
    \sin\left(\wt(\bar x) \frac{\uppi}{2}\right) \iverson{\forall i: x_i = 0}
    \nonumber\\
    &=
    \sin\left(\wt(\bar x) \frac{\uppi}{2}\right) \iverson{\wt(\bar x) = n}
    \nonumber\\
    &=
    \sin\left(\frac{n \uppi}{2}\right)\iverson{\wt(x) = 0} \nonumber\\
    &= \iverson{\wt(x) = 0} \times \begin{cases}  0, & \mbox{for even } n \\ 1, & \mbox{for } n \equiv 1 \pmod{4} \\ -1, & \mbox{for } n \equiv 3 \pmod{4} \end{cases}\nonumber\\
    &= 
    (-1)^{\frac{n-1}2} \iversontext{$n$ odd} \iverson{x = 0^n}.
\end{align}

Since $c_{A^n}(x)$ vanishes whenever $n$ is odd and $s_{A^n}(x)$ vanishes whenever $n$ is even, their product $t_{A^n}(x)$ must always vanish regardless of $n$, i.e., $t_{A^n}(x)=0$. Consequently, the probability expression in \cref{eq:probability_compact}
simplifies to
\begin{align}
    p_{A^n}(x) &= c_{A^n}^2(x) + s_{A^n}^2(x) \nonumber\\
    &= \left(
    (-1)^{n/2} \iversontext{$n$ even} \iverson{x = 1^n}
    \right)^2 + \left(
(-1)^{\frac{n-1}2} \iversontext{$n$ odd} \iverson{x = 0^n}
    \right)^2
    \nonumber\\
    &=
    \iversontext{$n$ even}\iverson{x=1^n} + \iversontext{$n$ odd}\iverson{x=0^n},
    \label{eq:probability_An_derivation}
\end{align}
which gives the expression in \cref{eq:probability_An}. In other words, when the number of players $n$ is even, then every player would volunteer with probability $1$, and when the number of players is odd, then every player would abstain with probability $1$.

Substituting \cref{eq:probability_An_derivation} into
\cref{eq:payoff_simplifed} gives
\begin{align}
    \$_a(A^n) &= 2\sum_{\substack{x \in \{0,1\}^n \\ x \neq 0^n \\ x_i = 0}} 
    \left(\iversontext{$n$ even} \underbrace{\iverson{x=1^n}}_{=0}  + \iversontext{$n$ odd}\underbrace{\iverson{x=0^n}}_{=0}
    \right) \nonumber\\
    &\quad
    + \sum_{k = 1}^n \left(2-\frac{1}{k}\right)\sum_{\substack{x \in \{0,1\}^n \\ x_i = 1 \\ \wt(x) = k}} \left(\iversontext{$n$ even}\iverson{x=1^n} + \iversontext{$n$ odd}\underbrace{\iverson{x=0^n}}_{=0}\right) \nonumber\\
    &=
    \iversontext{$n$ even} \sum_{k=1}^n \left(2-\frac 1k\right) \delta_{kn}
    \nonumber\\
    &=
    \left(2-\frac 1n\right) \iversontext{$n$ even},
\end{align}
which completes the derivation of \cref{eq:payoff_An}.

Next, to determine whether $A^n$ is a Nash equilibrium, we consider the scenario where a specific player $a \in [n]$ potentially deviates from the strategy $A$, while all other players $i \neq a$ continue to follow the strategy $A$. Under these conditions, the payoff for player $a$ may be calculated as follows:
\begin{align}
&c_{\theta, \phi}(x) \Big|_{(\theta_i, \phi_i)=A, \forall i \ne a} \nonumber\\&= \cos\left(\sum_{i:x_i=1} \phi_i\right) \prod_{i:x_i=0} \sin \frac{\theta_i}{2} \prod_{i:x_i=1} \cos\frac{\theta_i}{2} \Bigg|_{\theta_i=0, \phi_i=\frac \uppi 2, \forall i \ne a} \nonumber\\
&= \cos\left(x_a\phi_a+\sum_{i:x_i=1, i \ne a} \phi_i\right) \sin^{1-x_a}\left( \frac{\theta_a}2 \right)\cos^{x_a}\left( \frac{\theta_a}2\right)\prod_{\substack{i:x_i=0 \\ i \ne a}}\sin\frac{\theta_i}{2} \prod_{\substack{i:x_i=1\\ i \ne a}} \cos\frac{\theta_i}{2} \Bigg|_{
\theta_i =0, \phi_i = \tfrac\uppi 2, \forall i \neq a
}
\nonumber\\
&= \cos\left(x_a\phi_a+\sum_{i:x_i=1, i \ne a} \frac\uppi 2\right) \sin^{1-x_a}\left( \frac{\theta_a}2 \right)\cos^{x_a}\left( \frac{\theta_a}2\right) \iverson{\forall i \ne a, x_i =1} \nonumber\\ 
&= \cos\left(x_a\phi_a+(n-1)\frac{\uppi}2\right) \sin^{1-x_a}\left( \frac{\theta_a}2 \right)\cos^{x_a}\left( \frac{\theta_a}2\right) \iverson{\forall i \ne a, x_i =1} \nonumber\\
&= \left( (-1)^{\frac n2}\sin(x_a\phi_a)\iverson{n \text{ even}}  + (-1)^{\frac {n-1}2}\cos(x_a\phi_a)\iverson{n \text{ odd}} \right) \sin^{1-x_a}\left( \frac{\theta_a}2 \right)\cos^{x_a}\left( \frac{\theta_a}2\right) \iverson{\forall i \ne a, x_i =1} ,
\label{eq:c_strategy_A}
\end{align}
where the last line follows from the identity:
\begin{align}
    \cos\left(x\phi + (n-1)\frac \uppi 2\right) &=
    \begin{cases}
        \sin(x\phi), & \mbox{for } n \equiv 0 \pmod{4} \\
        \cos(x\phi), & \mbox{for } n \equiv 1 \pmod{4} \\
        -\sin(x\phi), & \mbox{for } n \equiv 2 \pmod{4} \\
        -\cos(x\phi), & \mbox{for } n \equiv 3 \pmod{4} \\
    \end{cases}
    \nonumber\\
    &= (-1)^{\frac n2}\sin(x\phi)\iverson{n \text{ even}}  + (-1)^{\frac{n-1}2}\cos(x\phi)\iverson{n \text{ odd}}.
\end{align}
Similarly,
\begin{align}
&s_{\theta, \phi}(x) \Big|_{(\theta_i, \phi_i)=A, \forall i \ne a} \nonumber\\&= \sin\left(\sum_{i:x_i=0} \phi_i\right) \prod_{i:x_i=0} \cos \frac{\theta_i}{2} \prod_{i:x_i=1} \sin \frac{\theta_i}{2} \Bigg|_{\theta_i=0, \phi_i=\frac \uppi 2, \forall i \ne a} \nonumber\\
&= \sin\left(\bar x_a\phi_a+\sum_{i:x_i=0, i \ne a} \phi_i\right) \cos^{1-x_a}\left( \frac{\theta_a}2 \right)\sin^{x_a}\left( \frac{\theta_a}2\right)\prod_{\substack{i:x_i=0 \\ i \ne a}}\cos\frac{\theta_i}{2} \prod_{\substack{i:x_i=1\\ i \ne a}} \sin\frac{\theta_i}{2} \Bigg|_{
\theta_i =0, \phi_i = \tfrac\uppi 2, \forall i \neq a
}
\nonumber\\
&= \sin\left(\bar x_a\phi_a+\sum_{i:x_i=0, i \ne a} \frac\uppi 2\right) \cos^{1-x_a}\left( \frac{\theta_a}2 \right)\sin^{x_a}\left( \frac{\theta_a}2\right) \iverson{\forall i \ne a, x_i =0} \nonumber\\ 
&= \sin\left(\bar x_a\phi_a+(n-1)\frac{\uppi}2\right) \cos^{1-x_a}\left( \frac{\theta_a}2 \right)\sin^{x_a}\left( \frac{\theta_a}2\right) \iverson{\forall i \ne a, x_i =0} \nonumber\\
&= \left( -(-1)^{\frac n2}\cos(\bar x_a\phi_a)\iverson{n \text{ even}}  + (-1)^{\frac{n-1}2}\sin(\bar x_a\phi_a)\iverson{n \text{ odd}} \right) \cos^{1-x_a}\left( \frac{\theta_a}2 \right)\sin^{x_a}\left( \frac{\theta_a}2\right) \iverson{\forall i \ne a, x_i =0} ,
\label{eq:s_strategy_A}
\end{align}
where the last line follows from the identity:
\begin{align}
    \sin\left(x\phi + (n-1)\frac \uppi 2\right) &=
    \begin{cases}
        -\cos(x\phi), & \mbox{for } n \equiv 0 \pmod{4} \\
        \sin(x\phi), & \mbox{for } n \equiv 1 \pmod{4} \\
        \cos(x\phi), & \mbox{for } n \equiv 2 \pmod{4} \\
        -\sin(x\phi), & \mbox{for } n \equiv 3 \pmod{4} \\
    \end{cases}
    \nonumber\\
    &= -(-1)^{\frac n2}\cos(x\phi)\iverson{n \text{ even}}  + (-1)^{\frac{n-1}2}\sin(x\phi)\iverson{n \text{ odd}}.
\end{align}

Since \cref{eq:c_strategy_A} and 
\cref{eq:s_strategy_A} are nonzero for disjoint sets of $x$, their product $t_{\theta, \phi}(x) |_{(\theta_i, \phi_i)=(0, \uppi/2), \forall i \ne a}=0$. Hence, the corresponding probability evaluates to
\begin{align}
& p_{\theta, \phi}(x) \Big|_{(\theta_i, \phi_i)=A, \forall i \ne a} = c_{\theta, \phi}(x)^2 \Big|_{(\theta_i, \phi_i)=A, \forall i \ne a} + s_{\theta, \phi}(x)^2 \Big|_{(\theta_i, \phi_i)=A, \forall i \ne a} \nonumber\\
&= \left(
\sin^2(x_a \phi_a) \iversontext{$n$ even} + \cos^2(x_a \phi_a) \iversontext{$n$ odd} 
\right)
\sin^{2(1-x_a)} \left(\frac{\theta_a}2\right) \cos^{2 x_a} \left(\frac{\theta_a}2\right) \iverson{\forall i \neq a, x_i = 1} \nonumber\\
&\quad+
\left( 
\cos^2(\bar x_a \phi_a) \iversontext{$n$ even} + \sin^2(\bar x_a \phi_a) \iversontext{$n$ odd} 
\right)
\cos^{2(1-x_a)} \left(\frac{\theta_a}2\right) \sin^{2 x_a} \left(\frac{\theta_a}2\right) \iverson{\forall i \neq a, x_i = 0}.
\end{align}

We will now consider the cases where $n$ is even and $n$ is odd separately. 

\begin{enumerate}
    \item When $n$ is even, the probability distribution simplifies to
\begin{align}
p_{\theta, \phi}(x)\Big|_{(\theta_i, \phi_i)=A, \forall i \ne a} &= \sin^2(x_a \phi_a)\sin^{2(1-x_a)}\left(\frac{\theta_a}2\right) \cos^{2x_a}\left(\frac{\theta_a}2 \right) \iverson{\forall i \ne a, x_i = 1} \nonumber\\
&\quad + \cos^2 (\bar{x}_a \phi_a) \cos^{2(1-x_a)} \left(\frac{\theta_a}2\right) 
\sin^{2x_a}\left(\frac{\theta_a}2\right)\iverson{\forall i \ne a, x_i = 0} \\
&= \begin{cases} \sin^2(\phi_a)  \cos^2\left(\frac{\theta_a}2\right) & \mbox{for } x=1^n \\
\cos^2(\phi_a) \cos^2\left(\frac{\theta_a}2\right) & \mbox{for }x=0^n \\
\sin^2\left(\frac{\theta_a}2\right) & \mbox{for } x_a=1, x_i=0, \forall i \ne a \\
0 & \text{otherwise}.
\end{cases}
\end{align}
Hence, the payoff for player $a$ evaluates to
\begin{align}
\$_a(\theta, \phi)\Big|_{(\theta_i, \phi_i)=A, \forall i \ne a} &= \sin^2 \left( \frac{\theta_a}2\right) + \left(2-\frac1n\right) \sin^2(\phi_a) \cos^2\left(\frac{\theta_a}2\right) \nonumber\\
& \le \sin^2 \left( \frac{\theta_a}2\right) + \left(2-\frac1n\right)  \cos^2\left(\frac{\theta_a}2 \right) \nonumber\\
& \le 2-\frac1n = \$_a(A^n).
\end{align}
Therefore, the strategy $A$ provides the optimal payoff for player $a$ when all other players also choose the strategy $A$. Since the same reasoning applies to every player $a$, the strategy profile $A^n$ is a Nash equilibrium.

\item When $n$ is odd, the probability distribution simplifies to
\begin{align}
p_{\theta, \phi}(x)\Big|_{(\theta_i, \phi_i)=A, \forall i \ne a}
&= \cos^2(x_a \phi_a)\sin^{2(1-x_a)}\left(\frac{\theta_a}2\right) \cos^{2x_a}\left(\frac{\theta_a}2\right) \iverson{\forall i \ne a, x_i = 1} \nonumber\\
&\quad + \sin^2 (\bar{x}_a \phi_a)  
\cos^{2(1-x_a)} \left(\frac{\theta_a}2\right)\sin^{2x_a}\left(\frac{\theta_a}2\right) \iverson{\forall i \ne a, x_i = 0} \nonumber\\
&= \begin{cases}  \sin^2(\frac {\theta_a} 2), & \mbox{for } x_a=0, x_i = 1, \forall i \ne a \\
\cos^2(\phi_a) \cos^2\left(\frac{\theta_a}2\right), & \mbox{for } x=1^n \\
\sin^2(\phi_a)\cos^2\left(\frac{\theta_a}2\right), & \mbox{for } x=0^n \\
0, & \text{otherwise}.
\end{cases}
\end{align}
Hence, the payoff for player $a$ evaluates to
\begin{align}
\$_a(\theta, \phi)\Big|_{(\theta_i, \phi_i)=A, \forall i \ne a} &= 2 \sin^2 \left( \frac{\theta_a}2\right) + \left(2-\frac1n\right) \cos^2(\phi_a) \cos^2\left(\frac{\theta_a}2\right).
\label{eq:payoff_a_odd}
\end{align}

Now, suppose that player $a$ chooses the parameters $\theta_a = \uppi$ and $\phi_a = 0$. Then
\cref{eq:payoff_a_odd} evaluates to
\begin{align}
\$_a(\theta, \phi)\Big|_{(\theta_i, \phi_i)=A, \forall i \ne a; \theta_a = \uppi, \phi_a = 0} &= 2,
\end{align}
which exceeds the payoff of $0 = \$_a(A^n)$ that they would receive by following strategy $A$. Hence, when $n$ is odd, $A^n$ is not a Nash equilibrium.
\end{enumerate}
By combining the results above, we conclude that $A^n$ is a Nash equilibrium if and only if $n$ is even.
\end{proof}

\subsection{Pareto optimality}

In \cref{sec:strategy_Q,sec:strategy_A}, we presented two distinct families of symmetric Nash equilibria. In this section, we will demonstrate that these Nash equilibria are also Pareto optimal, as defined earlier in \cref{def:pareto_optimality}. More broadly, we will prove that any strategy profile where every player receives a payoff of $2-\frac 1n$ must be Pareto optimal. This includes the strategy profile $Q^n$, as well as the strategy profile $A^{n}$ when $n$ is even, both of which satisfy this payoff condition.
\begin{theorem}
\label{thm:generalpareto}
    Let $s \in \Theta^n$ be any strategy profile for which $\$_a(s) = 2-\frac{1}{n}$ for all $a \in [n]$. Then $s$ is Pareto optimal in the $n$-player quantum volunteer's dilemma $G^{(n)}_{\mathrm{QVD}}$.
\end{theorem}

\begin{proof}
To prove that $s$ is Pareto optimal, it suffices to show that for any player $i \in[n]$ and any strategy profile $t \in \Theta^n$, if $\$_i(t) > \$_i(s)$, then there exists another player $j \in [n]\backslash\{i\}$ such that $\$_j(t) < \$_j(s)$.

To establish this, fix a player $i\in [n]$ and a strategy profile $t \in \Theta^n$ such that $\$_i (t) > \$_i(s)$, i.e.,
\begin{align}
    2-\frac 1n < \$_i(t).
\label{eq:pareto_optimality_i}
\end{align}
Suppose, for the sake of contradiction, that there does not exist any other player $j \in [n]\backslash\{i\}$ such that $\$_j(t) < \$_j(s)$. By de Morgan's laws for quantifiers, this assumption is equivalent to $\$_j(t) \geq \$_j(s)$ for all $j \in [n]\backslash \{i\}$, i.e.,
\begin{align}
    2 - \frac 1n \leq \$_j(t), \quad \forall j \in [n]\backslash\{i\}. \label{eq:pareto_optimality_j}
\end{align}

Taking the sum of \cref{eq:pareto_optimality_i} and the sum over all $j \neq i$ of \cref{eq:pareto_optimality_j}, we obtain
\begin{align}
    n\left(2-\frac{1}{n}\right) < \sum_{j=1}^n \$_j(t),
\end{align}
which implies that
\begin{align}
    2n - 1 &< \sum_{j=1}^n 
    \raisebox{-3mm}{$\left[\rule{0em}{11mm}\right.$}
    2\sum_{\substack{x \in \{0,1\}^n \\ x \neq 0^n \\ x_j = 0 }} p_t(x) + \sum_{k=1}^n \left(2-\frac{1}{k}\right) \sum_{\substack{x \in \{0,1\}^n \\ x_j = 1 \\ \wt(x) = k}} p_t(x)
    \raisebox{-3mm}{$\left.\rule{0em}{11mm}\right]$}
    \nonumber\\
&= 2 \underbrace{\sum_{j=1}^n \sum_{\substack{x \in \{0,1\}^n \\ x \neq 0^n \\ x_j = 0 }} p_t(x)}_{\circled{1}} + \underbrace{\sum_{k=1}^n \left(2-\frac{1}{k}\right) \sum_{j=1}^n \sum_{\substack{x \in \{0,1\}^n \\ x_j = 1 \\ \wt(x) = k}} p_t(x)}_{\circled{2}}.
\label{eq:pareto_optimality_intermediate_1}
\end{align}
Now, the first sum in \cref{eq:pareto_optimality_intermediate_1} simplifies to
\begin{align}
    \circled{1} &= 
     \sum_{j=1}^n \sum_{\substack{x \in \{0,1\}^n \\ x \neq 0^n \\ x_j = 0 }} p_t(x) =  \sum_{j=1}^n \sum_{\substack{x \in \{0,1\}^n \\ x \neq 0^n  }} \bar{x}_jp_t(x) = \sum_{\substack{x \in \{0,1\}^n \\ x \neq 0^n  }} \sum_{j=1}^n  \bar{x}_jp_t(x) 
    = \sum_{\substack{x \in \{0,1\}^n \\ x \neq 0^n  }} \wt(\bar{x})p_t(x) \nonumber\\
    &=  \sum_{k=1}^n \sum_{\substack{x \in \{0,1\}^n 
    \\ \wt(x)=k }} \wt(\bar{x})p_t(x) 
     =  \sum_{k=1}^n \sum_{\substack{x \in \{0,1\}^n \\ \wt(x)=k }} (n-k)p_t(x) 
     =  \sum_{k=1}^n (n-k) \sum_{\substack{x = \{ 0,1\}^n \\ \wt(x) = k}} p_t(x),
     \label{eq:circled1simplified}
\end{align}
and the second sum in \cref{eq:pareto_optimality_intermediate_1} simplifies to
\begin{align}
    \circled{2} &= 
    \sum_{k=1}^n \left(2-\frac{1}{k}\right) \sum_{j=1}^n \sum_{\substack{x \in \{0,1\}^n \\ x_j = 1 \\ \wt(x) = k}} p_t(x) =  \sum_{k=1}^n \left(2-\frac{1}{k}\right) \sum_{j=1}^n \sum_{\substack{x \in \{0,1\}^n  \\ \wt(x) = k}} x_j p_t(x) 
= \sum_{k=1}^n \left(2-\frac{1}{k}\right) \sum_{\substack{x \in \{0,1\}^n  \\ \wt(x) = k}} \sum_{j=1}^n  x_jp_t(x) 
\nonumber\\&= \sum_{k=1}^n \left(2-\frac{1}{k}\right) \sum_{\substack{x \in \{0,1\}^n  \\ \wt(x) = k}} \wt(x)p_t(x) =\sum_{k=1}^n (2k-1) \sum_{\substack{ x \in \{0,1\}^n \\ \wt(x) = k }} p_t(x). \label{eq:circled2simplified}
\end{align}
Substituting \cref{eq:circled1simplified,eq:circled2simplified} back into \cref{eq:pareto_optimality_intermediate_1} yields:    
\begin{align}
2n-1 &< 2\sum_{k=1}^n (n-k) \sum_{\substack{x = \{ 0,1\}^n \\ \wt(x) = k}} p_t(x) + \sum^n_{k=1} (2k-1) \sum_{\substack{ x \in \{0,1\}^n \\ \wt(x) = k }} p_t(x)
    \nonumber\\&= (2n-1) \sum^n_{k=1} \sum_{\substack{ x \in \{0,1\}^n \\ \wt(x) = k }}
    p_t(x) = (2n-1) \sum_{\substack{ x \in \{0,1\}^n \\ x \neq 0}} p_t(x)
    \nonumber\\ 
    &\leq (2n-1)\sum_{x \in \{0,1\}^n} p_t(x) 
    \nonumber\\ 
    &= 2n-1 ,
\end{align}
which is a contradiction, since $2n-1$ cannot be less than itself.

Therefore, there must exist another player $j \in [n]\backslash\{i\}$ for which $\$_j(t) < \$_j(s)$. This concludes the proof that $s$ is Pareto optimal. 
\end{proof}

A straightforward corollary of \cref{thm:generalpareto} is that those strategy profiles $Q^n$ and $A^n$ where each player's payoff is $2-1/n$ are Pareto optimal. We state this formally as follows:
\begin{corollary}
In the $n$-player quantum volunteer's dilemma $G^{(n)}_{\mathrm{QVD}}$, the following strategy profiles are Pareto optimal:
\begin{enumerate}
\item $Q^n$, for all $n \ge 2$,
\item $A^n$, for all even $n$.
    \end{enumerate}
\end{corollary}

\begin{proof}
From \cref{eq:payoff_Qn,eq:payoff_An}, any strategy profile $s \in \{Q^n:n\geq 2\} \cup \{A^n:n \mbox{ even}\}$ satisfies $\$_a(s) = 2-\frac1n$ for each player $a \in [n]$. Thus, by Theorem \ref{thm:generalpareto}, these strategy profiles are Pareto optimal.
\end{proof}

\section{Conclusion}
\label{sec:conclusion}

In this study, we explored a quantum generalization of the classical volunteer's dilemma by incorporating quantum strategies, utilizing the quantization framework introduced by Eisert, Wilkens, and Lewenstein \cite{eisert1999quantum}. Our analysis of the quantum volunteer's dilemma with multiple players reveals key insights into the strategic advantages of decision-making in a quantum context.

We derived explicit analytical expressions for the expected payoffs of players, demonstrating that the quantum version of the game offers crucial advantages over its classical counterpart. Specifically, the quantum volunteer's dilemma features symmetric Nash equilibria---$Q^n$ for $n\leq 9$ and $A^n$ for even $n$---which yield higher expected payoffs compared to the unique symmetric Nash equilibrium of the classical game, where players employ mixed strategies. Notably, these Nash equilibria we identified are also Pareto optimal.

Our findings contribute to a deeper theoretical understanding of how quantum mechanics can impact strategic interactions in game theory. By highlighting the benefits of quantum strategies, this work paves the way for exploring their practical implementations and potential applications in diverse real-world contexts.

\subsection*{Acknowledgements}

\small{D.E.K.\ acknowledges funding support from the National Research Foundation, Singapore and the Agency for Science, Technology and Research (A*STAR) under its Quantum Engineering Programme (NRF2021-QEP2-02-P03), A*STAR C230917003, and A*STAR under the Central Research Fund (CRF) Award for Use-Inspired Basic Research (UIBR).
}

\subsection*{Competing interests}
\small{
The authors declare no competing interests.}

\subsection*{Data Availability}
\small{
Data sharing is not applicable to this article as no datasets were generated or analyzed during the current study.}

\bibliographystyle{unsrt}
\bibliography{reference}
\end{document}